\documentclass[a4paper,english,cleveref,autoref,thm-restate]{lipics-v2019}

\usepackage{algorithm}
\usepackage[noend]{algpseudocode}

\usepackage{etoolbox}
\usepackage{tikz}

\usetikzlibrary{tikzmark}
\usetikzlibrary{calc}

\errorcontextlines\maxdimen

\newcommand{\ALGtikzmarkcolor}{black}
\newcommand{\ALGtikzmarkextraindent}{4pt}
\newcommand{\ALGtikzmarkverticaloffsetstart}{-.5ex}
\newcommand{\ALGtikzmarkverticaloffsetend}{-.5ex}
\makeatletter
\newcounter{ALG@tikzmark@tempcnta}

\newcommand\ALG@tikzmark@start{
    \global\let\ALG@tikzmark@last\ALG@tikzmark@starttext
    \expandafter\edef\csname ALG@tikzmark@\theALG@nested\endcsname{\theALG@tikzmark@tempcnta}
    \tikzmark{ALG@tikzmark@start@\csname ALG@tikzmark@\theALG@nested\endcsname}
    \addtocounter{ALG@tikzmark@tempcnta}{1}
}

\def\ALG@tikzmark@starttext{start}
\newcommand\ALG@tikzmark@end{
    \ifx\ALG@tikzmark@last\ALG@tikzmark@starttext
    \else
        \tikzmark{ALG@tikzmark@end@\csname ALG@tikzmark@\theALG@nested\endcsname}%
        \tikz[overlay,remember picture] \draw[\ALGtikzmarkcolor] let \p{S}=($(pic cs:ALG@tikzmark@start@\csname ALG@tikzmark@\theALG@nested\endcsname)+(\ALGtikzmarkextraindent,\ALGtikzmarkverticaloffsetstart)$), \p{E}=($(pic cs:ALG@tikzmark@end@\csname ALG@tikzmark@\theALG@nested\endcsname)+(\ALGtikzmarkextraindent,\ALGtikzmarkverticaloffsetend)$) in (\x{S},\y{S})--(\x{S},\y{E});
    \fi
    \gdef\ALG@tikzmark@last{end}
}

\apptocmd{\ALG@beginblock}{\ALG@tikzmark@start}{}{\errmessage{failed to patch}}
\pretocmd{\ALG@endblock}{\ALG@tikzmark@end}{}{\errmessage{failed to patch}}
\makeatother

\usepackage{graphicx}

\newcommand{\alg}{\text{ALG}}
\newcommand{\opt}{\text{OPT}}

\newcommand{\E}{\text{E}}
\newcommand{\Prob}{\text{Pr}}
\newcommand{\te}[1]{\text{#1}}
\newcommand{\dist}[2]{d(#1,#2)}

\newcommand{\DetAlg}{\texttt{PD-OMFLP}}
\newcommand{\RandAlg}{\texttt{RAND-OMFLP}}

\newcommand{\BigO}[1]{\mathcal{O}\left(#1\right)}

\newcommand{\SqrtS}{\sqrt{|S|}}
\newcommand{\lln}{\frac{\log n}{\log \log n}}

\DeclareMathOperator*{\argmax}{arg\,max}

\allowdisplaybreaks

\title{The Online Multi-Commodity Facility Location Problem\footnote{This work was partially supported by the German Research Foundation (DFG) within the Collaborative Research Centre On-The-Fly Computing (GZ: SFB 901/3) under the project number 160364472.}}

\author{Jannik Castenow}{Heinz Nixdorf Institute \& Computer Science Dept., Paderborn University, 33102 Paderborn, Germany}{jannik.castenow@upb.de}{https://orcid.org/0000-0002-8585-4181}{}

\author{Bj\"orn Feldkord}{Heinz Nixdorf Institute \& Computer Science Dept., Paderborn University, 33102 Paderborn, Germany}{bjoernf@hni.upb.de}{https://orcid.org/0000-0001-6591-2420}{}

\author{Till Knollmann}{Heinz Nixdorf Institute \& Computer Science Dept., Paderborn University, 33102 Paderborn, Germany}{tillk@mail.upb.de}{https://orcid.org/0000-0003-2014-4696}{}

\author{Manuel Malatyali}{Heinz Nixdorf Institute \& Computer Science Dept., Paderborn University, 33102 Paderborn, Germany}{manuel.malatyali@upb.de}{}{}

\author{Friedhelm Meyer auf der Heide}{Heinz Nixdorf Institute \& Computer Science Dept., Paderborn University, 33102 Paderborn, Germany}{fmadh@upb.de}{}{}

\authorrunning{Castenow, et al.}

\Copyright{Jannik C., Bj\"orn F., Till K. Manuel M., Friedhelm M.a.d.H.}


\keywords{Online Multi-Commodity Facility Location, Competitive Ratio, Online Optimization, Facility Location Problem}

\relatedversion{A conference version of this paper was accepted at the 32nd ACM Symposium on Parallelism in Algorithms and Architectures (SPAA 2020).}

\begin{document}

\nolinenumbers

\bibliographystyle{plain}

\maketitle

\begin{abstract}
  We consider a natural extension to the metric uncapacitated Facility Location Problem (FLP) in which requests ask for different commodities out of a finite set \( S \) of commodities.
  Ravi and Sinha (SODA 2004) introduced the model as the \emph{Multi-Commodity Facility Location Problem} (MFLP) and considered it an offline optimization problem.
  The model itself is similar to the FLP: i.e., requests are located at points of a finite metric space and the task of an algorithm is to construct facilities and assign requests to facilities while minimizing the construction cost and the sum over all assignment distances.
  In addition, requests and facilities are heterogeneous; they request or offer multiple commodities out of the set $S$.
  A request has to be connected to a set of facilities jointly offering the commodities demanded by it.
  In comparison to the FLP, an algorithm has to decide not only if and where to place facilities, but also which commodities to offer at each.

  To the best of our knowledge we are the first to study the problem in its online variant in which requests, their positions and their commodities are not known beforehand but revealed over time.
  We present results regarding the competitive ratio.
  On the one hand, we show that heterogeneity influences the competitive ratio by developing a lower bound on the competitive ratio for any randomized online algorithm of \( \Omega (  \sqrt{|S|} + \frac{\log n}{\log \log n}  ) \) that already holds for simple line metrics.
  Here, \( n \) is the number of requests.
  On the other side, we establish a deterministic \( \mathcal{O}(\sqrt{|S|} \cdot \log n) \)-competitive algorithm and a randomized \( \mathcal{O}(\sqrt{|S|} \cdot \frac{\log n}{\log \log n} ) \)-competitive algorithm for the problem.
  Further, we show that when considering a more special class of cost functions for the construction cost of a facility, the competitive ratio decreases given by our deterministic algorithm depending on the function.
\end{abstract}

\section{Introduction}\label{section:Introduction}

Consider the scenario of a provider of services in a network infrastructure.
Clients in the network might appear over time at locations in the network that are unknown to the provider and ask for a subset of the offered services.
For a scalable solution, the provider aims at placing instances of the required services close to the appearing requests to minimize the query cost for the requests.
When instantiating a service, there is typically a cost due to overhead for the set-up and the allocation of computational resources.
For example, such cost may be due to a virtual machine containing the service at the location of the instance.
It seems natural that it is worthwhile to offer a combination of services in a single virtual machine as opposed to instantiating each service on its own: i.e., the cost for instantiating a set of services increases less than linear with the number of offered services.
Additionally, a client that requests multiple services could benefit from communicating with a network node offering a subset of the requested services.
It is much cheaper to communicate with a single network node that offers multiple services than to communicate with different network nodes that serve the same set of services together.

The scenario above can be nicely modeled by extending the well-known \emph{Facility Location Problem (FLP)}.
Within the entire paper we assume the metric uncapacitated case if not mentioned otherwise.
In the metric Facility Location Problem, we are given requests located at points of a metric space and possible facility locations with the associated opening cost.
The task of an algorithm is to open facilities and connect each request to an open facility, while minimizing the total cost for opening facilities and the sum over all distances between requests and the facilities they are assigned to.

The natural extension of this problem is the \emph{Multi-Commodity Facility Location Problem (MFLP)} introduced in \cite{Ravi:2004:MFL:982792.982841}, in which each request asks for a subset of commodities out of a finite set.
Facilities are enabled to offer a subset of commodities when being opened and an algorithm has to ensure that a request is connected to a set of facilities jointly offering the requested commodities.
Facility costs are now determined not only by the location but also by the set of offered commodities.
The connection cost of a client is determined by the sum of distances to all facilities it is connected to.
This model generalizes the extensively studied FLP and introduces additional hardness, because an algorithm has to decide not only where to open facilities and how to connect the clients, but also which commodities to offer at a facility.

Our goal is to develop algorithms for the online version of the MFLP, which we call the OMFLP.
In the online variant, the requests are not known beforehand but revealed over time.
On arrival of a request, an algorithm has to immediately assign it to a set of facilities jointly offering the requested commodities.
Thereby, the algorithm has the possibility to open new facilities and determine the set of offered commodities for each newly opened facility.
Decisions on where to place a facility offering which commodities and how to connect a request are made irrevocably by the algorithm.
We analyze our online algorithms under the standard notion of the competitive ratio.

\begin{definition}
  [Competitive Ratio]
  Let \( P \) be a problem with a set of instances \( I \).
  Let \alg{} be an online algorithm and \opt{} be an optimal offline algorithm for \( P \).
  Denote by $\textnormal{Cost}(A, i)$ the total cost of an algorithm $A$ on an instance $i\in I$.
  Then \alg{} is called $c$-competitive if for all instances $i\in I$ it holds that
  $\textnormal{Cost}(\alg{}, i) \leq c\cdot \textnormal{Cost}(\opt{}, i) + a$ for some constant \( a \) independent of \( i \).
\end{definition}

\subsection{Model \& Problem Definition}\label{section:Introduction:Model-and-Problem-Definition}

We consider the metric non-uniform uncapacitated MFLP (Multi-Commodity Facility Location Problem
).
Here, we are given a metric space with point set \( M \) and a set \( R \) of \emph{requests} located at points of \( M \).
Each request \( r \in R \) demands a set \( s_{r} \subseteq S \) of \emph{commodities} out of a finite set \( S \).
The task of an algorithm is to compute a set of facilities \( F \) located at points of \( M \), determine for each facility which set of commodities is offered and then define an assignment of each request in \( R \) to a set of facilities in \( F \) while minimizing the sum of the \emph{construction cost} and the \emph{assignment cost}.
The algorithm is allowed to build multiple facilities on the same point.
Each request \( r \in R \) has to be connected to a set of facilities \( F'\subseteq F \) such that every commodity requested by \( r \) is offered by at least one facility in \( F' \).
We denote the \emph{distance} in the metric space of \( r \) to a facility at \( m \) by \( \dist{r}{m} \).
The connection cost for \( r \) is then determined by the sum of the distances from \( r \) to every facility of \( F' \).
Facilities of the algorithm are constructed with a \emph{configuration} \( \sigma \subseteq S \), i.e., a set of commodities offered at the facility.
Each facility in \( F \) induces a construction cost of \( f_{m}^{\sigma} \) where \( m \in M \) is the point where the facility is located and \( \sigma \subseteq S \) is the configuration of the facility.
Note that \( f_{m}^{\sigma} \) is given for each \( m \in M \) and each \( \sigma \subseteq S \) beforehand.

\noindent
\paragraph{Primal \& Dual Linear Program}
The following \emph{Integer Linear Program (ILP)} represents the MFLP.
\begin{align*}
  \min               & \sum_{m\in M} \sum_{\sigma \subseteq S} f_{m}^{\sigma} y_{m}^{\sigma}                                   + \sum_{m\in M} \sum_{\sigma \subseteq S} \sum_{r\in R} \sum_{s \subseteq s_{r}} \dist{m}{r} x_{mrs}^{\sigma} \\
  \textnormal{s.t. } & \sum_{m\in M}\sum_{\sigma \subseteq S} \sum_{s\subseteq \sigma: e\in s} x_{mrs}^{\sigma}  \geq 1
                     & \hspace*{-3.4cm}\forall r\in R, \forall e\in s_{r}                                                                                                                                                                    \\
                     & x_{mrs}^{\sigma}                                                                                                    \leq y_{m}^{\sigma}
                     & \hspace*{-3.4cm} \forall m\in M,\forall \sigma \subseteq S, \forall r\in R, \forall s \subseteq s_{r}                                                                                                                 \\
                     & x_{mrs}^{\sigma}, y_{m}^{\sigma}                                                                                    \in\{0,1\}
                     & \hspace*{-3.4cm} \forall m\in M,\forall \sigma \subseteq S, \forall r\in R, \forall s \subseteq s_{r}
\end{align*}

Here, \( y_{m}^{\sigma} \) represents a variable indicating that at \( m\in M \) there is a facility in configuration \( \sigma \subseteq S \).
\( x_{mrs}^{\sigma} \) indicates that the subset \( s \subseteq s_{r} \) of commodities requested by request \( r \) is served by a facility at \( m\in M \) in configuration \( \sigma \).
The first set of constraints ensure that every commodity of a request is served by a facility that \( r \) is connected to, while the second set of constraints ensure that requests are connected to and served only by facilities opened with a respective configuration.

Observe that, given fixed \( m, r, \sigma \), the connection cost for serving any subset \( s\subseteq s_{r} \cap \sigma \) by configuration \( \sigma \) at \( m \) is the same, namely \( \dist{m}{r} \).
Therefore, it is safe to assume that it is always better to tackle \( x_{mrs}^{\sigma} \) for maximal \( s \subseteq s_{r}\cap \sigma \), allowing us to eliminate explicitly reflecting \( s \) in \( x_{mrs}^{\sigma} \).
The ILP simplifies to:
\begin{align*}
  \min               & \sum_{m\in M} \sum_{\sigma \subseteq S} f_{m}^{\sigma} y_{m}^{\sigma}              + \sum_{m\in M} \sum_{\sigma \subseteq S} \sum_{r\in R} \dist{m}{r} x_{mr}^{\sigma}                                                                                           \\
  \textnormal{s.t. } & \sum_{m\in M}\sum_{\sigma \subseteq S: e\in \sigma}  x_{mr}^{\sigma} \geq 1                                                                                            & \hspace*{-1.75cm} \forall r\in R, \forall e\in s_{r}   \phantom{.}                      \\
                     & x_{mr}^{\sigma}                                                                                 \leq y_{m}^{\sigma}                                                    & \hspace*{-1.75cm} \forall m\in M,\forall \sigma \subseteq S, \forall r\in R \phantom{.} \\
                     & x_{mr}^{\sigma}, y_{m}^{\sigma}                                                                 \in\{0,1\}                                                             & \hspace*{-1.75cm} \forall m\in M,\forall \sigma \subseteq S, \forall r\in R
  .
\end{align*}

The corresponding dual is then as follows.
For convenience, define \( (a)_{+} : = \max\{a,0\} \) for any number \( a \) and \( z_{e}^{\sigma} = 1  \) if and only if \( e \in \sigma \).
\begin{align*}
  \max & \sum_{r\in R} \sum_{e\in s_{r}} a_{re}                                                                      &                                                                                                   \\
  \textnormal{s.t. }
       & \sum_{e\in s_{r}} a_{re} z_{e}^{\sigma} - \dist{r}{m} \leq b_{mr}^{\sigma}                                  & \hspace*{-0.5cm}\forall r \in R, \forall m\in M, \forall \sigma \subseteq S                       \\
       & \sum_{r\in R} b_{mr}^{\sigma} \leq f_{m}^{\sigma}                                                           & \hspace*{-0.5cm} \forall m\in M,\forall \sigma \subseteq S                                        \\
       & a_{re}, b_{mr}^{\sigma}                                                                              \geq 0 & \hspace*{-0.5cm}  \forall e \in s_{r}, \forall r\in R, \forall m\in M, \forall \sigma \subseteq S
\end{align*}

For \( z_{e}^{\sigma} = 0 \), \( - \dist{m}{r} \leq b_{mr}^{\sigma}  \) is tautological so the first set of constraints can be reduced to \( \left(\sum_{e\in s_{r} \cap \sigma} a_{re} - \dist{r}{m}\right)_{+} \leq b_{mr}^{\sigma} \).
Combined with the second set of constraints this yields a simplified dual as below.
\begin{align*}
  \max & \sum_{r\in R} \sum_{e\in s_{r}} a_{re}                                                                   &                                                          \\
  \textnormal{s.t. }
       & \sum_{r\in R} \left( \sum_{e \in s_{r} \cap \sigma} a_{re} - \dist{m}{r} \right)_{+} \leq f_{m}^{\sigma} & \hspace*{0.7cm}\forall m\in M,\forall \sigma \subseteq S \\
       & a_{re}                                                                               \geq 0              & \forall r\in R, \forall e \in s_{r}
\end{align*}

\noindent
\paragraph{Regarding the construction cost function}
For the construction cost function \( f_{m}^{\sigma} \), we would first like to observe that it can safely be assumed to be subadditive, i.e., for a fixed \( m \in M \) and any \( \sigma \subseteq S \) it holds for all \( a,b \subseteq \sigma \) with \( a \cup b = \sigma \) that
\begin{align*}
  f_{m}^{\sigma} \leq f_{m}^{a} + f_{m}^{b}
  .
\end{align*}
Assume that the cost function does not fulfill subadditivity.
Then for each \( \sigma, a, b \) and \( m \) violating the inequality as above, any algorithm that wants to cover the commodities of \( \sigma \) at \( m \) would simply not construct one facility with configuration \( \sigma \) but two facilities in configurations \( a \) and \( b \).
Thus, when considering the minimum possible construction cost for covering \( \sigma \) at \( m \), subadditivity is implied.

When looking at the literature in the offline case, we observe that the hardness of the problem when approximating it varies a lot depending on the allowed cost functions.
More specifically, a constant approximation is achievable when restricting the construction cost function.
Among others, it is assumed to be linear: i.e., \( f_{m}^{a\cup b} = f_{m}^{a} + f_{m}^{b} \) \cite{DBLP:conf/soda/ShmoysSL04}.
On the other hand, for general cost functions one cannot approximate better than by a factor of \( \Omega(\log |S|) \), due to a reduction from the weighted set cover problem \cite{Ravi:2004:MFL:982792.982841}.
The question arises if such a dependency on the cost function also exists in the online variant.
Naturally, as a starting point one could assume that the construction costs depend only on the number of commodities.
We losen this by demanding that our cost function fulfills
\begin{align}
  \forall \sigma \subseteq S, m\in M: &  & \frac{f_{m}^{\sigma}}{|\sigma|} \geq \frac{f_{m}^{S}}{|S|}
  \label{inequality:main-assumption}.\hspace*{2cm}
\end{align}
Condition~\ref{inequality:main-assumption} assures that the construction cost per commodity is minimal when considering \( S \) entirely.
Keeping in mind that the construction cost increases less than linearly in the number of included commodities, this seems reasonable.
Note, that assuming a cost function that depends only on the number of offered commodities together with the always present subadditivity implies Condition~\ref{inequality:main-assumption} but is not equivalent to it, i.e., our assumption is strictly more general.

In our lower bound, we will see that \emph{prediction} on \( S \) is needed: i.e., an algorithm has to offer types at facilities which were not yet requested.
Mainly, Condition~\ref{inequality:main-assumption} allows us to simplify the decision on which commodities to predict at a fixed point.
In \cref{section:outlook} we discuss how we could drop our assumption for future research.

\noindent
\paragraph{A different cost model}
We would like to briefly note that one could also formulate a different model for the MFLP.
Assume that a request \( r \) is served multiple commodities by a single facility at \( m \in M \).
In our model, the connection cost of \( r \) to \( m \) is counted only once.
This reflects the idea that multiple commodities are served by a single communication path (incurring cost).
One could argue that the connection cost should be counted separately per commodity of \( r \) that is served by the facility.
This model can be easily simulated in our model by replacing each request with \( s_{r}\subseteq S \) by \( |s_{r}| \) many requests demanding a single commodity.
Note that this possibly increases the sequence length by a factor of at most \( |S| \) in the online case.
However, it seems reasonable that the number of commodities \( |S| \) is polynomial in the number of requests \( n \) such that the competitive ratios of our algorithms increase only by a factor of \( 2 \).

\noindent
\paragraph{Additional notation}
We usually suppress the time in our notation to improve readability.
Note, however, that the set \( F \) of facilities opened by the algorithm as well as the set \( R \) of requests changes as time goes.
For convenience, let \( F(e) \subseteq F \) for a commodity \( e \in S \) be the set of facilities that are currently open offering \( e \).
Similarly, at a fixed point in time, let \( \hat{F} \subseteq F \) be the set of currently open facilities offering all commodities in $S$ and let \( R(e) \subseteq R \) be the set of requests in \( R \) that request \( e \in S \).
For a given request \( r \) and a commodity \( e\in s_{r} \), denote by \( \dist{F(e)}{r} \) the distance of \( r \) to the closest open facility offering \( e \).

\subsection{Related Work}\label{section:Introduction:Related-Work}

Facility location problems have long been of great interest for economists and computer scientists.
In this overview, we focus only on provable results for metric variants.

A comprehensive overview of different techniques used for approximation algorithms for the metric Facility Location Problem can be found in~\cite{DBLP:conf/approx/Shmoys00}.
The currently best approximation ratio for the problem is roughly~1.488~\cite{DBLP:journals/iandc/Li13} and a lower bound of 1.463 holds in case $NP\notin DTIME(n^{\BigO{\log\log n}})$~\cite{DBLP:journals/jal/GuhaK99}.
For our work, the primal-dual algorithm by Jain and Vazirani~\cite{DBLP:journals/jacm/JainV01} is particularly interesting, as it inspired algorithms for variants of facility location such as the online~\cite{Fotakis:2007:PAO:1224558.1224672} and the leasing variant~\cite{Nagarajan2013,DBLP:conf/sirocco/KlingHP12}, which in turn heavily influence our deterministic algorithm.

For the Online Facility Location problem, Meyerson~\cite{Meyerson2005} introduced a randomized algorithm which he argued was $\BigO{\log n}$-competitive.
He also showed a non-constant lower bound on the competitive ratio.
The algorithm was later shown to be $\BigO{\frac{\log n}{\log\log n}}$-competitive, which Fotakis~\cite{Fotakis2008} showed to be the best possible competitive ratio for any online algorithm. He also gave a deterministic algorithm with the same competitive ratio, resolving the question of the competitive ratio for Online Facility Location up to a constant.

Fotakis~\cite{Fotakis:2007:PAO:1224558.1224672} also provided a simpler online algorithm with a slightly worse asymptotic competitive ratio, but which runs much more efficiently and admits smaller constants in the analysis.
This algorithm was used to derive an algorithm for the leasing variant as well~\cite{Nagarajan2013} and we also use it as a basis for the approach to our model.
With regard to the competitive ratio, it should be noted that it was shown that Meyerson's algorithm in~\cite{Meyerson2005} performs much better if the scenario is not strictly adversarial.
In fact, gradually weakening the power of the adversary to influence the order of requests also decreases the competitive ratio~\cite{DBLP:conf/soda/Lang18}.
Allowing the algorithm to make small corrections to the position of its facilities also brings the competitive ratio down to a constant~\cite{DBLP:conf/spaa/FeldkordH18}.
This also motivates us to give a variant of Meyerson's algorithm for our model, which then naturally benefits from the same phenomena in non-adversarial scenarios or where decisions are not completely irreversible.

The first work on our model in the offline case was made by Ravi and Sinha~\cite{DBLP:journals/siamdm/RaviS10}, who constructed an $\BigO{\log |S|}$ approximation and showed that this cannot be improved by more than a constant by using a reduction from the weighted set cover problem.
The only restriction on the cost function is that it needs to be subadditive: $f_{m}^{a\cup b} \leq f_{m}^{a} + f_{m}^{b}$.
Shmoys et al.~\cite{DBLP:conf/soda/ShmoysSL04} showed that a constant approximation ratio can be achieved when the cost function is more restricted to be linear ($f_{m}^{a \cup b} = f_{m}^{a} + f_{m}^{b}$) and additionally ordered on the potential facility locations: i.e., between two facility locations all commodities are more expensive on one location than on the other.
Fleischer~\cite{DBLP:conf/aaim/FleischerLTZ06} considered the problem in the non-metric variant and showed an approximation ratio logarithmic in the number of requests, facility locations and commodities for both the capacitated und uncapacitated case.
Svitkina and Tardos~\cite{DBLP:journals/talg/SvitkinaT10} gave a constant approximation for hierarchical cost functions: i.e., opening costs are modeled by a tree with the requests as leaves, where the cost is derived by summing up the cost of connecting the respective requests to the root.
Finally, Poplawski and Rajaraman~\cite{DBLP:conf/soda/PoplawskiR11} considered approximation algorithms for a variant in which the requests are served by connecting them to a set of facilities covering all the commodities via a Steiner tree.

To the best of our knowledge, no work has been done on an online variant of our problem so far.

\subsection{Our Results}\label{section:Introduction:Outline}

The rest of this paper is structured as follows.
In \cref{section:lower-bounds}, we show that the competitive ratio depends on the size of $S$ by deriving a general lower bound of \( \Omega(\SqrtS + \lln) \) even for randomized algorithms on simple metric spaces such as the line and when assuming a cost function that depends only on the number of offered commodities.
Observe that it is trivial to achieve an algorithm having a competitive ratio of \( \mathcal{O}(|S|\cdot\lln) \) simply by solving an instance of the OFLP for each commodity separately, using Fotakis' algorithm~\cite{Fotakis2008}, for example.
Our main result is a deterministic algorithm for the problem achieving a competitive ratio of \( \mathcal{O}(\SqrtS\cdot\log n) \) in \cref{section:deterministic-algorithm}.
The algorithm is based on the primal dual algorithm by Fotakis \cite{Fotakis:2007:PAO:1224558.1224672}, but now has to incorporate the choice of the set of commodities for each facility.
Interestingly, the algorithm distinguishes only between facilities that serve a single commodity and facilities serving all commodities.

In our general analysis, we only assume Condition~\ref{inequality:main-assumption}, concerning the construction cost function \( f_{m}^{\sigma} \).
When restricting the instances to more specific construction cost functions, e.g., polynomials in the size of the set of offered commodities, we are able to show an adaptive lower bound as well as an improvement in the competitive ratio of the deterministic algorithm, both depending on the parameters of the cost function.
We elaborate on this in \cref{section:deterministic-algorithm:Improved-Bounds}.

In \cref{section:Randomized-Algorithm}, we complement our result by a randomized algorithm, achieving an expected competitive ratio of \( \mathcal{O}(\SqrtS\cdot\lln) \).
Our randomized algorithm achieves a slightly better competitive ratio than the deterministic approach and is much more efficient to implement.
The algorithm is based on Meyerson's randomized algorithm for the FLP \cite{Meyerson2005} and uses similar adaptions in the analysis as utilized in the deterministic algorithm.

\section{Lower Bounds}\label{section:lower-bounds}

Due to \cite{Fotakis2008}, we know that the lower bound on the competitive ratio for the OFLP is \( \Theta(\frac{\log n}{\log \log n}) \) even on line metrics.
Next, we prove a lower bound for the OMFLP that includes the total number of commodities \( |S| \) as captured in \cref{theorem:lower-bound}.
Combining both lower bounds yields the result presented in \cref{corollary:lower-bound}.
We start by presenting the proof of the lower bound.
Afterwards, we explain how the lower bound motivates the fundamental design decision for our algorithms so that they distinguish only between facilities serving a single commodity and facilities serving all commodities.

\begin{theorem}
  \label{theorem:lower-bound}
  No randomized online algorithm for the Online Multi-Com\-modity Facility Location problem can achieve a competitive ratio better than \( \Omega \left(\sqrt{|S|}\right) \), even on a single point.
\end{theorem}

\begin{corollary}\label{corollary:lower-bound}
  No randomized online algorithm for the Online Multi-Com\-modity Facility Location problem can achieve a competitive ratio better than \( \Omega \left(\sqrt{|S|} + \frac{\log n}{\log \log n}\right) \), even on a line metric.
\end{corollary}

\begin{proof}[Proof of \cref{theorem:lower-bound}]
  According to Yao's Principle~\cite{YaosPrinciple} (see e.g., \cite[Chapter 8]{DBLP:books/daglib/0097013} for details), it is sufficient to construct a probability distribution over demand sequences for which the expected ratio between the costs of the deterministic online algorithm performing best against the distribution and the optimal cost is \(\Omega (\sqrt{|S|}) \).
  We prove the theorem by first defining a suitable function for the facility opening costs.
  By \alg{}, we denote a deterministic online algorithm and by \opt{} an optimal offline algorithm for the OMFLP.
  To improve readability, we assume that \( \sqrt{|S|}\in \mathbb{N} \).
  We consider a single point \( m \in M \) of a given metric space.

  \paragraph{Facility opening costs.}
  Let the cost function for facilities for the point \( m\) be \( g(|\sigma|) = f_{m}^{\sigma} = \left\lceil \frac{|\sigma|}{\sqrt{|S|}} \right\rceil  \): i.e., the cost depends only on the size of the configuration.

  \paragraph{Sequence definition.}
  Consider a set \( S'\subset S \) of the size \( |S'| = \sqrt{|S|} \) containing randomly selected commodities.
  The commodities are selected uniformly and independently of each other.
  One at a time, trigger a single request demanding a single commodity in \( S' \) not yet requested at \( m \).

  \paragraph{Competitive ratio.}
  An optimal algorithm builds a single facility serving all the commodities in \( S' \) at \( m \).
  Hence, \opt{} pays no more than a total of \(  g(\sqrt{|S|}) = 1\).

  Contrary to \opt{}, \alg{} does not know the set \( S' \) until it has been revealed after \( |S'| = \sqrt{|S|} \) requests.
  In each time step, \alg{} has to serve the commodity being requested and can additionally buy commodities to cover potential future requests.
  Observe that if \alg{} does not predict, i.e., it only includes commodities that were already requested when building a facility, it builds \( \SqrtS \) facilities for a total price of \( \SqrtS \).

  We observe that \alg{} constructs facilities in \emph{rounds}, where in the \( i \)-th round a not yet covered commodity \( s\in S' \) is requested and \alg{} builds a facility serving \( s \) and \( t_{i} \) additional commodities.
  \( t_{i} \) is entirely chosen by \alg{} and some of the additionally covered commodities might be requested later.
  Note that we may assume that \alg{} does not build new facilities when a commodity that is already covered is requested.
  We can simply move whatever \alg{} then buys to the next round, and the rounds in which nothing happens can be removed from the following calculation.
  Let \( X \) be the number of rounds needed by \alg{}.
  Then the cost of \alg{} is determined by both \( X \) and \( T:=\sum_{i} t_{i} \), because \alg{} builds \( X \) facilities and covers \( T \) many commodities in total: i.e., the cost of \alg{} is at least $\max\{X,T/\sqrt{|S|}\}$.
  Consider \cref{figure:lower-bound:experiment} for a depiction.
  If \( X \geq \SqrtS/2 \), \alg{}'s competitive ratio is at least \( \Omega(\SqrtS) \).
  Therefore, assume that \( X < \SqrtS/2 \).
  Next, we show that in this case \( T \) is large on expectation.
  \begin{figure}[htb]
    \centering
    \includegraphics[width= 0.47\textwidth,page=2, trim=3.8cm 13cm 15.25cm 2.4cm, clip=true]{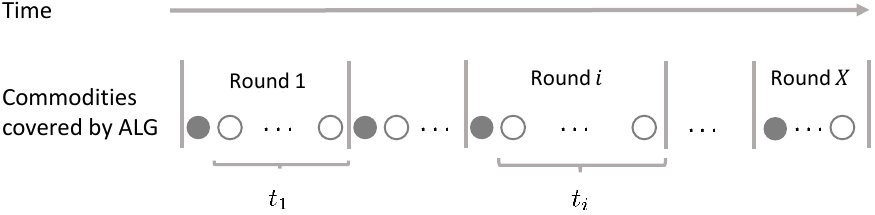}
    \caption{\alg{}'s behavior in rounds \( 1,\dots, X \). In round \( i \), a not yet covered commodity (filled) is requested and covered by a facility of \alg{}. Thereby, \alg{} covers \( t_{i} \) additional commodities. \alg{}'s cost is determined by \( X \) as well as \( T:=\sum_{i} t_{i} \), because it builds \( X \) facilities and covers at least \( T \) commodities.}
    \label{figure:lower-bound:experiment}
  \end{figure}

  Let \( S'_{a} \subset S' \) be the set of commodities that are not covered by \alg{} before they are requested: i.e., they are requested but not \emph{predicted} by \alg{}.
  Similarly, let \( S'_{b} = S' \setminus S'_{a} \)  be the set of commodities that are requested and predicted by \alg{}.
  Observe that \( |S'_{a}|=X < \SqrtS / 2  \) and \( |S'_{b}| \geq \SqrtS /2 \).
  Let \( S_{b} = S \setminus S'_{a} \) be the total set of commodities out of which \alg{} predicts in total \( T \) many, including the ones in \( S'_{b} \).
  Then \( |S_{b}| = |S| - |S'_{a}| \geq |S| - \SqrtS /2 \geq |S|/2 \).
  We are interested in bounding \( T \).
  Since the commodities of \( S_{b} \) are indistinguishable for \alg{} and they all have the same probability of being chosen for \( S' \) (unknown to \alg{}), \alg{}'s decision on which commodities are predicted can be viewed as arbitrary and independent of the chosen \( S'_{b} \).
  Thus, it is equivalent to model \alg{}'s selection by assuming that \alg{} draws \( T \) times without replacement out of the set \( S_{b} \) and covers all commodities of \( S'_{b} \).
  Then the expected number of draws \( \E[T] \) until \( S'_{b} \) is covered can be seen as
  \begin{align}
    \E[T]
     & = \sum_{i = |S'_{b}|}^{|S_{b}|} \Prob[T = i] \cdot i
    \geq \sum_{i=|S|/c}^{|S|/2} \Prob[T = i] \cdot \frac{|S|}{c}
    \nonumber                                                        \\
     & = \Prob\left[T \geq \frac{|S|}{c} \right] \cdot \frac{|S|}{c}
    \geq \Prob\left[T > \frac{|S|}{c} \right] \cdot \frac{|S|}{c}
    \label{inequality:lower-bound-expectation}
  \end{align}
  where \( c \geq 4 \) is a sufficiently large constant.
  Next, we show that with constant probability \( |S|/c \) draws are not sufficient to cover \( S'_{b} \): i.e., \( T > |S|/c \).

  Assume that we draw exactly \( |S|/c \) many times out of \( |S|/2 \) commodities of which \( \SqrtS / 2 \) ones are requested.
  Let \( Y \) be the number of drawn requested commodities.
  Then \( Y \) is hypergeometrically distributed (\( Y \sim \text{Hypergeometric}(|S|/2,\SqrtS / 2,|S|/c) \)) with mean \( \E[Y] = \SqrtS / c \).
  In case \( Y < \SqrtS / 2 \), not all commodities of \( S'_{b} \) are covered: i.e., \( T > |S|/c \).
  Since \( Y \) is hypergeometrically distributed we can apply the bounds of \cite{doi:10.1080/01621459.1963.10500830,CHVATAL1979285} and get
  \begin{align}
     & \Prob\left[T > \frac{|S|}{c} \right]
    = \Prob\left[Y < \frac{\SqrtS}{2}\right]
    = 1 - \Prob\left[Y \geq \frac{\SqrtS}{2}\right]
    \nonumber                                                                                 \\
     & = 1 - \Prob\left[Y \geq \frac{\SqrtS}{c} + \frac{\SqrtS}{2} - \frac{\SqrtS}{c} \right]
    \nonumber                                                                                 \\
     & = 1 - \Prob\left[Y \geq \E[Y] + \frac{\SqrtS\,(c-2)}{2\,c} \right]
    \nonumber                                                                                 \\
     & = 1 - \Prob\left[Y \geq \E[Y] + \frac{(c-2)}{2\SqrtS}\, \frac{|S|}{c} \right]
    \nonumber                                                                                 \\
     & \geq 1 - \te{e}^{-2\frac{(c-2)^{2}}{4\, |S|} \frac{|S|}{c}}
    = 1 - \te{e}^{-\frac{(c-2)^{2}}{2\,c}}
    \geq 1 - \te{e}^{-\frac{1}{2}}
    \geq \frac{1}{4}
    \label{inequality:lower-bound-prob-of-linear}
    .
  \end{align}
  Combining \cref{inequality:lower-bound-expectation} and \cref{inequality:lower-bound-prob-of-linear} yields
  \begin{align}
    \E[T]
    \geq \frac{|S|}{16}
    \label{inequality:lower-bound-expectation-in-S}
    .
  \end{align}
  Therefore, the expected cost for \alg{} is at least
  \begin{align*}
    \max\left\{X,\,g(\E[T])\right\}
    = \max\left\{X,\,\frac{\E[T]}{\SqrtS}\right\}
    \geq \frac{\SqrtS}{16}
    .
  \end{align*}
  Recapitulate that \opt{}'s cost is \( 1 \) and the theorem holds true.
\end{proof}

Our lower bound motivates the usage of \emph{prediction}.
Any algorithm that aims at achieving a competitive ratio depending on \( |S| \) by less than a linear factor has to offer commodities that were not yet requested at some point.
Otherwise, one can easily force it to build \( \Omega (|S|) \) facilities while \opt{} needs only a single one combining all necessary commodities for a cost that is a \( 1/|S| \) fraction of the algorithm's cost (with the choice of a suitable cost function).

When introducing prediction, it is unclear \emph{how} to choose the commodities that are offered while not yet requested.
In our lower bound we can see that a single rule helps us to simplify this decision significantly.
A simple way to have a tight bound against the lower bound on a single point is to construct only facilities serving a single commodity until \( \SqrtS \) many facilities have been constructed.
Afterwards, directly build a facility serving \emph{all} commodities.
Intuitively, do not predict as long as it is not worthwhile and if it is, cover everything.
When all commodities are covered, \opt{} has to cover at least \( \SqrtS \) commodities, which yields a competitive ratio of \( \mathcal{O}(\SqrtS) \) due to Condition~\ref{inequality:main-assumption}.

We denote facilities serving a single commodity as \emph{small facilities} and facilities serving all commodities as \emph{large facilities}.
Both of our algorithms are based on deciding between small and large facilities.
The main difficulty now is to incorporate the aforementioned prediction into a general metric and to establish a suitable threshold that dictates \emph{when} the algorithm switches from building small facilities to building a large one.

\section{A deterministic Algorithm}\label{section:deterministic-algorithm}

In the following section we present our deterministic algorithm for the OMFLP.
As motivated in the previous section, the algorithm considers only the construction of small and large facilities.

\subsection{Algorithm}\label{section:deterministic-algorithm:Algorithm}

Our algorithm \DetAlg{} (\cref{algorithm:deterministic-algorithm}) is shown below.
It is inspired by the primal dual formulation of Fotakis' deterministic algorithm \cite{Fotakis:2007:PAO:1224558.1224672} for the OFLP presented in \cite{Nagarajan2013}, which achieves a competitive ratio of \( \mathcal{O}(\log n) \).
\DetAlg{} achieves a competitive ratio of \( \mathcal{O}(\SqrtS \log n) \).
In its core, our algorithm uses the dual variables of each commodity that a request demands as an investment.
This investment is paid towards connecting to existing small/large facilities (see Constraints (1) and (2) below) as well as towards the construction of and the connection to new small/large facilities (see Constraints (3) and (4) below).
Thereby, all commodities demanded by a request invest together into the connection to or the construction of a large facility, because they all profit by having one shared connection.

Next, we present details on the investment phase.
Consider the following four constraints for a given request \( r \) with commodity set \( s_{r} \).
Our algorithm \DetAlg{} guarantees that the constraints always hold during its execution.

\begin{enumerate}
  \item \( a_{re} \leq \dist{F(e)}{r} \) for all \( e \in s_{r} \)
  \item \( \sum_{e\in s_{r}} a_{re} \leq \dist{\hat{F}}{r} \)
  \item \( (a_{re} - \dist{m}{r})_{+} + \\\sum_{j\in R: e\in s_{j}} (\min\{a_{je}, \dist{F(e)}{j}\} - \dist{m}{j})_{+} \leq f_{m}^{\{e\}} \)\\for all \( e\in s_{r} \)
  \item \( \left(\sum_{e\in s_{r}} a_{re} - \dist{m}{r}\right)_{+} + \\\sum_{j \in R} \left( \min \left\{\sum_{e\in s_{j}} a_{j e}, \dist{\hat{F}}{j} \right\} - \dist{m}{j} \right)_{+} \leq f_{m}^{S} \)
\end{enumerate}

Note that all sets and distances are taken with respect to the current time step.

\begin{algorithm}[htb]
  \caption{\DetAlg{} on arrival of request \( r \) with set \( s_{r} \)}
  \label{algorithm:deterministic-algorithm}
  \begin{algorithmic}[1]
    \While{Not all \( e\in s_{r} \) are served}
    \State Simultaneously increase all \( a_{re} \) for every \( e \in s_{r} \) that is
    \Statex \hspace*{0.7cm} not yet served.
    \If{Constraint (1) or (3) is tight for \( e\in s_{r} \) and \( m \in M \)}
    \State Freeze \( a_{re} \) and declare \( e \) served by \( m \).
    \State In case of Constraint (3), declare facility serving \( e \) at
    \Statex \hspace*{1.2cm} \( m \) temporary open.
    \EndIf
    \If{Constraint (2) or (4) is tight for \( m\in M \)}
    \State For all \( e'\in s_{r} \), freeze \( a_{re'} \) and declare \( e' \) served by \( m \).
    \State Remove all temporarily open facilities.
    \State In case of Constraint (4), open a large facility at \( m \).
    \EndIf
    \State Open all remaining temporary open facilities.
    \EndWhile
  \end{algorithmic}
\end{algorithm}

Requests that appeared earlier than the current one reinvest into small facilities exactly what they invested earlier, as can be seen in Constraint (3).
The decision on \emph{when} to open the first large facility implicitly depends on how much investment has been made towards small facilities.
This can be seen in the minimum term of Constraint (4).
For the first large facility in a certain area, at most the total investment of all requests for which a large facility would be worthwhile is invested.
Thereby, a facility becomes worthwhile if the investment minus the distance to the facility is greater than zero.
After the first large facility is established in an area, the investment of a request's commodity into a new facility is also bounded by the distance of the closest large facility to it.
In this way, only the initial investment is reinvested into large facilities in total.

\subsection{Analysis}\label{section:deterministic-algorithm:Analysis}

Next, we analyze the competitive ratio of our algorithm.
For this, we proceed very similar to \cite{Nagarajan2013}: i.e., we first show that the primal solution of the algorithm is bounded by the sum of all dual variables (\cref{section:deterministic-algorithm:Analysis:Bounded-Cost}) and then prove that an appropriate scaling of the dual variables leads to a feasible dual solution (\cref{section:deterministic-algorithm:Analysis:feasible-dual}).
By weak duality the competitive ratio of the algorithm is then bounded by the used scaling factor, resulting in the correctness of \cref{theorem:PD-extend:algorithm-competitive}.
For the second part of the analysis in \cref{section:deterministic-algorithm:Analysis:feasible-dual} we need to bound the objective function of a special class of weighted set cover problems.
We introduce the respective problem definition and show an upper bound on the total weight needed for a cover in \cref{section:deterministic-algorithm:Ordered-Covering}.
In the entire proof, let the scaling factor be \( \gamma = 1/(5 \sqrt{|S|} H_{n}) \), where \( H_{n} = \sum_{k=1}^{n} \frac{1}{k} \) is the \( n \)-th harmonic number.

\begin{theorem}\label{theorem:PD-extend:algorithm-competitive}
  \textnormal{\DetAlg{}} has a competitive ratio of
  \begin{align*}
    \BigO{\SqrtS \cdot \log n}
    .
  \end{align*}
\end{theorem}

\subsubsection{Bounding the algorithm's cost}\label{section:deterministic-algorithm:Analysis:Bounded-Cost}\(  \)

\noindent
Our main goal here is to show \cref{corollary:PD-extend:cost-of-algorithm-bounded-by-dual-vars}: i.e., the cost of the algorithm is bounded by the sum of all duals.
The proofs of the following lemmas are close to the proof of Lemma 4.1 in \cite{Nagarajan2013}, yet we have to carefully distinguish between small and large facilities and the respective investment of the requests.

\begin{lemma}\label{lemma:PD-extend:cost-of-assignment-bounded}
  The assignment cost of our algorithm's solution is bounded by \( \sum_{r\in R} \sum_{e\in s_{r}} a_{re} \).
\end{lemma}

\begin{proof}
  For a request \( r \) it holds that either (i) all commodities of \( s_{r} \) are assigned to only small facilities or (ii) all commodities of \( s_{r} \) are assigned to a single large facility.

  In (i), for a fixed commodity \( e\in s_{r} \) either Constraint (1) or (3) was true.
  For Constraint 1) \( a_{re} = \dist{F(e)}{r} \) and for Constraint 3) \( a_{re} \geq \dist{m}{r} \) for the point \( m \) to which \( e \) is assigned.
  Therefore, the connection cost is bounded in \( a_{re} \) for each \( e\in s_{r} \).

  In (ii), either Constraint (2) or (4) is true.
  In any case the complete connection cost for request \( r \) is bounded by \( \sum_{e\in s_{r}} a_{re}  \).
\end{proof}

\begin{lemma}\label{lemma:PD-extend:cost-of-small-facilities-bounded}
  The construction cost for small facilities of our algorithm's solution is bounded by \( \sum_{r\in R} \sum_{e\in s_{r}} a_{re} \).
\end{lemma}

\begin{proof}
  Throughout the proof we consider only a request's bid towards small facilities at all points.
  A request's bid is the contribution of the respective term in the sum of Constraints (1) or (3).
  We can ignore large facilities here, since their construction only reduces the bid of requests towards small facilities.

  Fix a commodity \( e\in S \) and consider only the small facilities offering \( e \).
  Observe that when a small facility is opened, its construction cost is bounded by the sum of all bids of requests for \( e \) (Constraint (3)).
  Any request \( r \) bids at most \( a_{re} \) towards any point \( m \in M \) due to the minimum term in Constraint (3).
  We show that when the bid of a request \( r \) is used to open a facility at \( m \in M \), all outstanding bids of \( r \) for other facilities serving \( e \) are reduced by the amount \( r \) bids for \( e \) towards \( m \).

  Assume that there are two locations \( m, m' \) without a small facility.
  Before commodity \( e \) of \( r \) is assigned, it bids \( (a_{re} - \dist{r}{m})_{+} \) and \( (a_{re} - \dist{m'}{r})_{+} \) towards both locations.
  Assume that a facility at \( m \) opens and commodity \( e \) of \( r \) is assigned to it.
  Then the bid of \( r \) for \( e \) towards \( m' \) reduces to \( (\dist{m}{r} - \dist{m'}{r})_{+} \).
  Thus, it was reduced by \( (a_{re} - \dist{m'}{r})_{+} - (\dist{m}{r} - \dist{m'}{r})_{+} = (a_{re} - \dist{m}{r}) \) which is the amount spent for the facility at \( m \).

  Assume that commodity \( e \) of \( r \) is assigned to a facility and \( m,m' \) are locations without a facility.
  When a facility at \( m \) opens and the bid of \( r \) for \( e \) reduces, it reduces by \( (\min\{a_{re}, \dist{F(e)}{r}\} - \dist{m}{r})_{+}  \) which is greater than \( 0 \): i.e., \( m \) is closer to \( r \) than any already open facility offering \( e \) and \( \dist{m}{r} \leq a_{re} \).
  We will show that the bid of \( r \) for \( e \) at \( m' \) reduces by exactly this amount.
  Once a small facility for \( e \) is opened at \( m \), \( \min\{a_{re}, \dist{F'(e)}{r}\} = \dist{F'(e)}{r} = \dist{m}{r} \) (where \( F' \) denotes the new facility set containing \( m \)).
  As a side note, when \( a_{re} \geq \dist{F(e)}{r} \) holds once, it will hold for all future configurations since we do not delete facilities.
  The bid \( r \) spends for \( e \) towards \( m' \) reduces by
  \begin{align*}
     & (\min\{a_{re}, \dist{F(e)}{r}\} - \dist{m'}{r})_{+}              \\
     & \phantom{=} - (\min\{a_{re}, \dist{F'(e)}{r}\}-\dist{m'}{r})_{+} \\
     & = (\min\{a_{re}, \dist{F(e)r}\}-\dist{m'}{r})_{+}                \\
     & \phantom{=} - (\dist{m}{r} - \dist{m'}{r})                       \\
     & = (\min\{a_{re}, \dist{F(e)}{r}\}-\dist{m}{r})_{+}
    .
  \end{align*}
\end{proof}

The proof of \cref{lemma:PD-extend:cost-of-large-facilities-bounded} is very close to the proof of \cref{lemma:PD-extend:cost-of-small-facilities-bounded}.

\begin{restatable}{lemma}{lemmaCostForLargeFacilitiesBounded}\label{lemma:PD-extend:cost-of-large-facilities-bounded}
  The construction cost for large facilities of our algorithm's solution is bounded by \(  \sum_{r\in R } \sum_{e\in s_{r}} a_{re} \).
\end{restatable}

\begin{proof}
  Observe that when a large facility is opened, its construction cost is bounded by the sum of all bids of requests (Constraint (4)).
  Any request \( r \) bids only at most \( \sum_{e\in s_{r}} a_{re} \) towards a large facility at any point \( m \in M \), due to the minimum term in Constraint (4).
  We show that when the bid of a request \( r \) is used to open a large facility at \( m \in M \), all outstanding bids of \( r \) for other large facilities are reduced by the amount \( r \) bids towards \( m \), similar to the case of the small facilities.

  Assume that there are two locations \( m, m' \) without a large facility.
  Before \( r \) is assigned, it bids \( (\sum_{e\in s_{r} } a_{re} - \dist{r}{m})_{+} \) and \( (\sum_{e\in s_{r} } a_{re} - \dist{m'}{r})_{+} \) towards both locations.
  Assume that a large facility at \( m \) opens and \( r \) is assigned to it.
  Then the bid of \( r \) towards the large facility at \( m' \) reduces to \( (\dist{m}{r} - \dist{m'}{r})_{+} \).
  Thus, it was reduced by \( ( \sum_{e\in s_{r} } a_{re} - \dist{m'}{r})_{+} - (\dist{m}{r} - \dist{m'}{r})_{+} = ( \sum_{e\in s_{r} } a_{re} - \dist{m}{r}) \), which is the amount spent for the large facility at \( m \).

  Assume that \( r \) is already assigned to a facility and \( m,m' \) are locations without a large facility.
  When a large facility at \( m \) opens and the bid of \( r \) reduces, it reduces by \( (\min \{ \sum_{e\in s_{r} } a_{re}, \dist{\hat{F}}{r}\} - \dist{m}{r})_{+}  \), which is greater than \( 0 \): i.e., \( m \) is closer to \( r \) than any already open large facility and \( \dist{m}{r} < \sum_{e\in s_{r} } a_{re} \).
  We will show that the bid of \( r \) for a large facility at \( m' \) reduces by exactly this amount.
  Once the large facility at \( m \) is opened, \( \min \{ \sum_{e\in s_{r} } a_{re}, \dist{\hat{F}'}{r}\} = \dist{\hat{F}'}{r} = \dist{m}{r} \) (where \( \hat{F}' \) denotes the new facility set containing large facilities including only the new one at \( m \)).
  As a side note, if \( \sum_{e\in s_{r} } a_{re} > \dist{\hat{F}}{r} \) holds once, it will hold for all future configurations since we do not delete facilities.
  The bid \( r \) spends towards a large facility at \( m' \) reduces by
  \begin{align*}
     & (\min \{\sum_{e\in s_{r} } a_{re}, \dist{\hat{F}}{r}\} - \dist{m'}{r})_{+}                \\
     & \phantom{=} - (\min \{\sum_{e\in s_{r} } a_{re}, \dist{\hat{F}'}{r}\} - \dist{m'}{r})_{+} \\
     & = (\min \{\sum_{e\in s_{r} } a_{re}, \dist{\hat{F}}{r}\} - \dist{m'}{r})_{+}              \\
     & \phantom{=} - (\dist{m}{r} - \dist{m'}{r})_{+}                                            \\
     & = \min \{\sum_{e\in s_{r} } a_{re}, \dist{\hat{F}}{r}\} - \dist{m}{r}
    .
  \end{align*}
\end{proof}

\begin{corollary}\label{corollary:PD-extend:cost-of-algorithm-bounded-by-dual-vars}
  The cost of the algorithm's solution is bounded by \( 3 \sum_{r \in R} \sum_{e\in s_{r} } a_{re}  \).
\end{corollary}

\subsubsection{\texorpdfstring{\( c \)}{c}-ordered covering}\label{section:deterministic-algorithm:Ordered-Covering}\(  \)

\noindent
Before we continue with the analysis, we introduce a special class of the weighted set cover problem.
A good solution to instances of this class is needed in the proofs of \cref{lemma:PD-extend:feasability-for-medium-configurations} and \cref{lemma:PD-extend:feasability-for-configurations}.
Our instances are defined below and we aim at finding a minimal weight covering of the set \( \{1,\dots,n\} \).

\begin{definition}[\( c \)-ordered covering]
  Consider elements \( 1, \dots, n \) and a given parameter \( c \geq 1 \).
  An instance for \( c \)-ordered covering is given as follows.
  For element \( i \), define \( A_{i} \subseteq \{ 1, \dots, i-1 \} \) and \( B_{i} \subseteq \{ 1, \dots, i-1\} \) such that \( A_{i} \cap B_{i} = \emptyset \) and \( A_{i} \cup B_{i} = \{1,\dots, i-1\} \).
  For any two elements \( i \) and \( j \) with \( i < j \) it holds \( B_{i} \subseteq B_{j} \).
  For every \( i = 1,\dots,n \) let there be a set \( \{i\} \) with weight \( \frac{c}{|B_{i}|+1} \) and a set \( \{i\} \cup A_{i} \) with weight \( c \).
\end{definition}

We will show that a covering with a weight of at most \( 2 c H_{n} \) can always be achieved.
For this, let us introduce some notation.
We call a set of elements \( \{i, \dots, j\} \subseteq \{1,\dots, n\} \) with \(  i\leq j \) of maximum cardinality a \emph{block} if \( B_{i} = B_{j} \).
For convenience, we say an element \( i \) copes the elements in \( A_{i} \).
Note that within a block, the \( B_{i} \) do not change.
Thus, each element copes all the previous elements in its block and possibly more elements.

Our proof consists of the following two steps:
\begin{enumerate}
  \item Given a \( c \)-ordered covering instance of length \( n \), we can cover \( x > 1 \) elements with a total weight of \( 2 c\,\sum_{i=n-x}^{n} \frac{1}{i} \).
  \item Given a \( c \)-ordered covering instance of length \( n \), the \( x \) previously covered elements can safely be removed from the instance and we can create a new ordered covering instance of length \( n - x \).
\end{enumerate}
It directly follows that the set \( \{1,\dots,n\} \) can be covered by a \( c \)-ordered covering instance with a weight of \( 2 c H_{n} \).

\begin{lemma}\label{lemma:PD:Subproblem:Covering-with-harmonic-weight}
  Given a \( c \)-ordered covering instance of length \( n \), we can cover \( x \geq 1 \) elements with a total weight of \( 2 c\, \sum_{i=n-x}^{n} \frac{1}{i} \).
\end{lemma}

\begin{proof}
  Consider the following two choices that cover at least the elements of the last block.
  \begin{enumerate}
    \item Select the set \( \{n\} \cup A_{n} \) with a weight of \( c \).
    \item For every element \( i \) of the last block, select the set \( \{i\} \) with a weight of \( c/(|B_{n}|+1) \) each.
  \end{enumerate}
  Observe that the element \( n \) copes \( n-|B_{n}| \) elements.
  Hence, the weight per coped element in case 1 is \( c/(n-|B_{n}|) \).
  Depending on which choice is cheaper per element, select one of the two choices.
  Now, the weight per selected element is bounded by
  \begin{align*}
    \min \left\{ \frac{c}{n-|B_{n}|}, \frac{c}{|B_{n}|+1} \right\} \leq \frac{2\,c}{n}.
  \end{align*}
  Assume \( x \) elements were covered.
  Then the total weight for the covered elements is
  \begin{align*}
    \sum_{i=n-x}^{n} \frac{2\,c}{n} \leq 2 c\, \sum_{i=n-x}^{n} \frac{1}{i}.
  \end{align*}
\end{proof}

\begin{lemma}\label{lemma:PD:Subproblem:Instance-can-be-rebuildt}
  Given a \( c \)-ordered covering instance of length \( n \), the element \( n \) and \( x \geq 0 \) arbitrary elements that are coped by it can be removed from the instance.
  We can transform the remaining instance into a new \( c \)-ordered covering instance of length \( n - x - 1 \).
\end{lemma}

\begin{proof}
  Observe that the element \( n \) can safely be removed from the instance by simply deleting the sets \( \{ n \} \) and \( \{ n \} \cup A_{n} \).

  Any other element \( i \) that is coped by \( n \) is not in any \( B_{j} \) for all \( j = 1, \dots, n \).
  Removing \( i \) (including the sets \( \{i\} \) and \( \{i\} \cup A_{i} \)) thus does not influence any \( B_{j} \), such that the following still holds:
  \begin{itemize}
    \item The weights of all remaining sets are untouched.
    \item For all remaining \( j \): \( A_{j} \cap B_{j} = \emptyset \).
    \item For all remaining elements \( j \) and \( k \) with \( j < k \): \( B_{j} \subseteq B_{k} \).
  \end{itemize}
  The condition that for all remaining \( j \) it has to hold \( A_{j} \cup B_{j} = \{1,\dots,j-1\} \) is violated due to the removal of \( i \).
  However, it can easily be fixed by consistently renaming every element \( j > i \) to \( j - 1 \).
  The resulting instance is a \( c \)-ordered covering instance not containing \( i \).

  The described procedure can be repeated for \( x > 1 \) arbitrary elements coped by \( n \), resulting in a \( c \)-ordered covering instance of length \( n - 1 - x \).
\end{proof}

\begin{lemma}\label{corollary:PD:Subproblem:covering-works-with-some-weight}
  The set \( \{1,\dots,n\} \) can be covered by a \( c \)-ordered covering instance with a weight of \( 2 c H_{n} \).
\end{lemma}

\begin{proof}
  By \cref{lemma:PD:Subproblem:Covering-with-harmonic-weight}, we can cover \( x \) elements of a \( c \)-ordered covering instance of length \( n \) with a weight of \( 2 c \, \sum_{i=n-x}^{n} \frac{1}{i} \).
  The covered elements can safely be removed by \cref{lemma:PD:Subproblem:Instance-can-be-rebuildt} since all covered elements are coped by the last element.
  This yields a \( c \)-ordered covering instance of length \( n - x \).
  Repeatedly applying \cref{lemma:PD:Subproblem:Covering-with-harmonic-weight} and \cref{lemma:PD:Subproblem:Instance-can-be-rebuildt} yields a covering of \( \{1,\dots, n\} \) with a weight of \( 2 c \, \sum_{i=1}^{n} \frac{1}{i} = 2 c H_{n} \).
\end{proof}

\subsubsection{A feasible dual solution}\label{section:deterministic-algorithm:Analysis:feasible-dual}\(  \)

\noindent
Next, we are ready to show that scaling down all \( a_{re} \) by \( \gamma \) leads to a feasible solution to the dual.
First, the distance of the nearest facility to a request demanding a specific commodity can be bounded as follows.

\begin{lemma}\label{lemma:PD-extend:same-commodity-distance-bounds-dual}
  Fix a commodity \( e\in S \).
  Consider two requests \( j, \ell \) which arrived at time \( t,t' \) where \( t < t' \) with \( e \in s_{j} \) and \( e \in s_{\ell} \).
  Let \( G \) be a set of facilities such that each facility serves \( s \).
  It holds \( \dist{G}{j} - \dist{m}{j} \geq a_{\ell e} - \dist{m}{ \ell} - 2 \dist{m}{j} \) at the time when we increase \( a_{\ell e} \).
\end{lemma}

\begin{proof}
  This proof is very similar to the proof of Lemma 4.2 in \cite{Nagarajan2013}.
  For completeness we restate it here in compliance to our notation.
  Note, that we have to carefully consider the commodities.

  Consider the facility \( g \in G \) closest to \( j \) when we increase the dual of \( \ell e \).
  The dual value \( a_{\ell e} \) is no more than \( \dist{g}{\ell} \) since \( g \) is open, serves the commodity of \( \ell \) and we could have assigned \( \ell e \) to \( g \).
  By the triangle inequality it holds \( a_{\ell e} \leq \dist{g}{\ell} \leq \dist{g}{j} + \dist{j}{m} + \dist{m}{\ell} \).
  At the time we increase \( a_{\ell e} \) it holds \( \dist{G}{j} - \dist{m}{j} = \dist{g}{j} - \dist{m}{j} \) and every facility of \( G \) offers the commodity of \( j \).
  Together this yields \( \dist{G}{j} - \dist{m}{j} = \dist{g}{j} - \dist{m}{j} \geq a_{ \ell e} - \dist{m}{\ell} - 2 \dist{m}{j} \).
\end{proof}

Note that as the set \(G\) we usually take the set of all large facilities when considering different commodities of \(j\) and \(\ell\) and the set of all large facilities and small facilities that serve \( e \).
In the remainder, we will prove that all constraints of the dual hold when using the variables \(  a_{re} \) as set by the algorithm, scaled by $\gamma$.

\begin{lemma}[Feasability for configurations \( 1 \leq |\sigma|\leq \sqrt{|S|} \)]\label{lemma:PD-extend:feasability-for-medium-configurations}
  Fix a configuration \( \sigma \) with \( |\sigma| \leq \sqrt{|S|} \).
  For any \( R'\subseteq R \) and any facility serving \( \sigma \) at \( m \in M \): \( \sum_{r\in R'} (\sum_{e\in s_{r}\cap \sigma } \gamma a_{re} - \dist{m}{r})_{+} \leq f_{m}^{\sigma} \).
\end{lemma}

\begin{proof}
  First, consider a single commodity \( s\in \sigma \).
  Consider any request \( \ell \in R' \) with \( s\in s_{\ell} \) at the time at which we increase \( a_{\ell s} \).
  Also, consider only those requests in \( R' \) that arrived earlier than \( \ell \).
  All other requests do not influence the dual variable \( a_{\ell s}  \).
  Due to Constraint (3), for \( \ell \) and \( s \) it holds that:
  \begin{align*}
    f_{m}^{\{s\}} \geq (a_{\ell s} - \dist{m}{\ell})_{+} + \hspace*{-0.2cm}\sum_{j\in R': s\in s_{j}}\hspace*{-0.2cm} \left(\min\{a_{js}, \dist{F(s)}{j}\} - \dist{m}{j} \right)_{+}.
  \end{align*}
  Let \( A_{\ell} \) be the set of requests of \( R' \) for which \( \min\{a_{js}, \dist{F(s)}{j}\} = a_{js} \) and, similarly, let \( B_{\ell} \) be the set of requests of \( R' \) for which \( \min\{a_{js}, \dist{F(s)}{j}\} = \dist{F(s)}{j} \) at the arrival of \( \ell \).

  For requests in \( B_{\ell} \) we can apply \cref{lemma:PD-extend:same-commodity-distance-bounds-dual} since all considered facilities serve the commodity \( s \).
  This yields:
  \begin{align*}
    f_{m}^{\{s\}} & \geq (a_{\ell s} - \dist{m}{\ell})_{+} + \sum_{j\in A_{\ell}} (a_{js} - \dist{m}{j})_{+}                                                           \\
                  & \phantom{\geq} + \sum_{j\in B_{\ell}} (\dist{F(s)}{j} - \dist{m}{j})_{+}                                                                           \\
                  & \geq (a_{\ell s} - \dist{m}{\ell}) + \sum_{j\in A_{\ell}} (a_{js} - \dist{m}{j})                                                                   \\
                  & \phantom{\geq} + \sum_{j\in B_{\ell}} (a_{\ell s} - \dist{m}{\ell} - 2 \dist{m}{j})                                                                \\
                  & \hspace*{-0.5cm} = (|B_{\ell}|+1)(a_{\ell s} - \dist{m}{\ell}) + \sum_{j\in A_{\ell}} (a_{js} - \dist{m}{j}) - 2 \sum_{j\in B_{\ell}} \dist{m}{j}.
  \end{align*}
  Denote \( 2 \sum_{j\in B_{\ell}} \dist{m}{j} \) by \( \lambda \).
  The inequality above implies the following two inequalities:
  \begin{align}
    (a_{\ell s} - \dist{m}{\ell})
     & \leq \frac{f_{m}^{\{s\}} + \lambda }{|B_{\ell}| + 1 } \label{inequality:proof-medium-config-inequality-one} \\
    (a_{\ell s} - \dist{m}{\ell}) + \sum_{j\in A_{\ell}} (a_{j s} - \dist{m}{j})
     & \leq f_{m}^{\{s\}} + \lambda
    .\label{inequality:proof-medium-config-inequality-two}
  \end{align}

  Now we model the task to bound \(X := \sum_{r\in R'} (a_{r s}-\dist{m}{r})_{+} \) by solving the problem of covering all the \( (a_{r s} - \dist{m}{r} )_{+} \) given an instance of \( c \)-ordered covering.
  The idea behind this is the following.
  Each time we cover an element \( (a_{r s} - \dist{m}{r} )_{+} \) of \( X \), we do so by applying either \cref{inequality:proof-medium-config-inequality-one} or \cref{inequality:proof-medium-config-inequality-two}.
  In case we apply \cref{inequality:proof-medium-config-inequality-one}, we remove one element \( (a_{r s} - \dist{m}{r} )_{+} \) from \( X \) and add a weight of \( \frac{f_{m}^{\{s\}} + \lambda }{|B_{\ell}| + 1 } \).
  In the other case of \cref{inequality:proof-medium-config-inequality-two}, we remove multiple elements from the sum of \( X \) and add a weight of \( f_{m}^{\{s\}} + \lambda \).
  We ask ourselves how much weight is achieved when removing every element of \( X \).
  The resulting weight then directly represents an upper bound for \( X \).

  Next, we define an instance of \( c \)-ordered covering based on inequalities (1) and (2).
  Our instance is as follows:
  Number the requests of \( R' \) from \( 1 \) to \( |R'(s)| \) in the order of arrival.
  The elements of our instance are \( 1,\dots,|R'(s)| \).
  Consider element \( i \).
  It represents \( (a_{r s} - \dist{m}{r})_{+} \) of the \( i \)-th arriving request \( r \) of \( R'(s) \).
  The sets \( B_{i} \) and \( A_{i} \) are given by the \( B_{r} \) and \( A_{r} \) as defined above.
  The parameter \( c \) of our ordered covering is \( f_{m}^{\{s\}} + \lambda \).
  For every element \( i \) there is a set \( \{i\} \) of weight \( c/(|B_{i}|+1) \) and a set \( \{i\} \cup A_{i} \) of weight \( c \).
  Notice that the weights of the sets correspond to \cref{inequality:proof-medium-config-inequality-one} and \cref{inequality:proof-medium-config-inequality-two}, respectively.

  Now, we show that this is a proper \( c \)-ordered covering instance.
  For any element \( i \), \( A_{i} \cap B_{i} = \emptyset \) by definition and \( A_{i} \cup B_{i} = \{1,\dots,i-1\} \), because exactly the requests of \( R'(s) \) that arrived earlier than the request corresponding to \( i \) have a defined value for \( \min\{a_{j s}, \dist{F(s)}{j}\} \).
  If a request \( r \) is in some \( A_{i} \) and in \( B_{i+1} \), it contributed to building a large facility of which the distance to itself is less than \( a_{r} \).
  Thus, for all following elements \( j > i \), \( r \) will stay in \( B_{j} \).
  In other words, for any two elements \( i, j \) with \( i < j \) it holds that \( B_{i} \subseteq B_{j} \).

  By \cref{corollary:PD:Subproblem:covering-works-with-some-weight} we know that \( \{(a_{r s} - \dist{m}{r} )_{+} \,|\, r\in R'(s)\} \) can be covered with a total weight of \( 2 (f_{m}^{\{s\}} + \lambda) H_{n} \).
  Each time an element is covered, this corresponds to applying either \cref{inequality:proof-medium-config-inequality-one} or \cref{inequality:proof-medium-config-inequality-two} to the respective \( (a_{r s} - \dist{m}{r})_{+} \) term, indicated by the increase in the weight of the covering.
  Note, that for any request \( r \in R'\setminus R'(s) \), commodity \( s \) is not requested and thus \( (a_{r s} - \dist{m}{r})_{+} = 0 \).
  Thus, we conclude:
  \begin{align*}
     & \sum_{r\in R'} (a_{r s} - \dist{m}{r})_{+} \leq 2 (f_{m}^{\{s\}} + \lambda) H_{n}         \\
     & \hspace*{2.65cm} \leq 2 f_{m}^{\{s\}} H_{n} + 4 H_{n} \sum_{r \in R'} \dist{m}{r}         \\
     & \Rightarrow \sum_{r\in R'} (a_{r s} - 5 H_{n} \dist{m}{r})_{+} \leq 2 f_{m}^{\{s\}} H_{n}
  \end{align*}
  Let \( s := \argmax_{e\in \sigma} f_{m}^{\{e\}} \).
  Now, by applying the inequality for each commodity of \( \sigma \) separately, we have
  \begin{align*}
     &                                     & \sum_{s\in \sigma} \sum_{r\in R'} (a_{r s} - 5 H_{n} \dist{m}{r})_{+}
     & \leq 2 f_{m}^{\{s\}} H_{n} |\sigma|                                                                                                                   \\
     & \Rightarrow                         & \sum_{r\in R'} \sum_{s\in \sigma \cap s_{r}} (a_{r s} - 5 H_{n} \dist{m}{r})_{+}
     & \leq 2 f_{m}^{\{s\}} H_{n} |\sigma|                                                                                                                   \\
     & \Rightarrow                         & \sum_{r\in R'} \left( \sum_{e\in \sigma \cap s_{r}} \frac{a_{r e}}{5 \sqrt{|S|} H_{n}} - \dist{m}{r}\right)_{+}
     & \leq \frac{2}{5} f_{m}^{\{s\}}                                                                                                                        \\
     & \Rightarrow                         & \sum_{r\in R'} \left( \sum_{e\in \sigma \cap s_{r}} \gamma a_{r e} -  \dist{m}{r}\right)_{+}
     & \leq f_{m}^{\sigma}
  \end{align*}
\end{proof}

Next, we approach configurations of a size of at least \( \SqrtS \).
The proofs of \cref{lemma:PD-extend:all-commodities-distance-bounds-dual} and \cref{lemma:PD-extend:feasability-for-configurations} are very similar to the proofs of \cref{lemma:PD-extend:same-commodity-distance-bounds-dual} and \cref{lemma:PD-extend:feasability-for-medium-configurations}, respectively.

What follows is a technical lemma similar to \cref{lemma:PD-extend:same-commodity-distance-bounds-dual}, which considers the distance of a request to the nearest large facility.

\begin{restatable}{lemma}{technicalLemmaDistanceLarge}\label{lemma:PD-extend:all-commodities-distance-bounds-dual}
  Consider two requests \( j, \ell \), which arrived at time \( t,t' \) where \( t < t' \).
  It holds that \( \dist{\hat{F}}{j} - \dist{m}{j} \geq \sum_{e\in s_{\ell}} a_{\ell e} - \dist{m}{\ell} - 2 \dist{m}{j} \) at the time when we increase the dual variables of \( \ell \).
\end{restatable}

\begin{proof}
  This proof is analogous to the proof of \cref{lemma:PD-extend:same-commodity-distance-bounds-dual}.
  Consider the facility \( g \in \hat{F} \) closest to \( j \) when we increase \( X:=\sum_{e\in s_{\ell}} a_{\ell e} \).
  The value of \( X \) is no more than \( \dist{g}{\ell} \) since \( g \) is open and a large facility and we could have assigned \( \ell \) completely to \( g \) (Constraint (2)).
  By the triangle inequality it holds that \( X \leq \dist{g}{\ell} \leq \dist{g}{j} + \dist{j}{m} + \dist{m}{\ell} \).
  At the time we increase \( X \) it holds that \( \dist{G}{j} - \dist{m}{j} = \dist{g}{j} - \dist{m}{j} \) and the facilities of \( \hat{F} \) can all serve \( j \) completely.
  Together this yields \( \dist{G}{j} - \dist{m}{j} = \dist{g}{j} - \dist{m}{j} \geq X - \dist{m}{\ell} - 2 \dist{m}{j} = \sum_{e\in s_{\ell}} a_{\ell e} - \dist{m}{\ell} - 2 \dist{m}{j} \).
\end{proof}

\begin{restatable}[Feasability for configurations \( |\sigma|>\sqrt{|S|} \)]{lemma}{fesabilityForLargeConfigurations}\label{lemma:PD-extend:feasability-for-configurations}
  Fix a configuration \( \sigma \)  with \( |\sigma| > \sqrt{|S|} \).
  For any \( R' \subseteq R \) and any facility serving \( \sigma \) at \( m \in M \) it holds that
  \\
  \( \sum_{r\in R'} \left( \sum_{e\in s_{r}\cap \sigma} \gamma a_{r e} - \dist{m}{r}\right)_{+} \leq f_{m}^{\sigma} \).
\end{restatable}

\begin{proof}
  Consider any request \( \ell \in R' \) and assume the time at which we increase \( \sum_{e\in s_{\ell}} a_{\ell e} \).
  Recapitulate that due to Constraint (4) for request \( \ell \):
  \begin{align*}
     & \left(\sum_{e\in s_{\ell}} a_{\ell e} - \dist{m}{\ell}\right)_{+} \hspace*{-0.1cm} + \sum_{j\in R'} \left(\min\left\{ \sum_{e \in s_{j} \cap \sigma} a_{je}, \dist{\hat{F}}{j} \right\} - \dist{m}{j}\right)_{+} \\
     & \leq \left( \sum_{e\in s_{\ell}} a_{\ell e} - \dist{m}{\ell}\right)_{+} \hspace*{-0.1cm} + \sum_{j\in R'} \left(\min\left\{ \sum_{e \in s_{j}} a_{je}, \dist{\hat{F}}{j} \right\} - \dist{m}{j}\right)_{+}       \\
     & \leq f_{m}^{S}
    .
  \end{align*}
  Let \( A_{\ell} \) be the set of requests of \( R' \) where \( \min\{\sum_{e\in s_{j}} a_{je}, \dist{\hat{F}}{j}\} = \sum_{e\in s_{j}} a_{je} \) and, similarly, let \( B_{\ell} \) be the set of requests of \( R' \) for which \( \min\{\sum_{e\in s_{j}} a_{je}, \dist{\hat{F}}{j}\} = \dist{\hat{F}}{j} \) at the arrival of \( \ell \).

  For the requests in \( B_{\ell} \) and their commodities, we can apply \cref{lemma:PD-extend:all-commodities-distance-bounds-dual}.
  Thus,
  \begin{align*}
    f_{m}^{S}
     & \geq \left( \sum_{e\in s_{\ell}} a_{\ell e} - \dist{m}{\ell}\right) + \sum_{j\in A_{\ell}} \left(\sum_{e \in s_{j}} a_{je} - \dist{m}{j}\right)                \\
     & \phantom{\geq} + \sum_{j\in B_{\ell}} \left( \dist{\hat{F}}{j} - \dist{m}{j}\right)                                                                            \\
     & \geq \left( \sum_{e\in s_{\ell}} a_{\ell e} - \dist{m}{\ell}\right) + \sum_{j\in A_{\ell}} \left(\sum_{e \in s_{j}} a_{je} - \dist{m}{j}\right)                \\
     & \phantom{\geq} + \sum_{j\in B_{\ell}} \left( \sum_{e\in s_{\ell}} a_{\ell e} - \dist{m}{\ell} - 2 \dist{m}{j} \right)
    \\
     & \geq (|B_{\ell}|+1)\left( \sum_{e\in s_{\ell e}} a_{\ell e}-\dist{m}{\ell} \right) + \sum_{j\in A_{\ell}} \left(\sum_{e \in s_{j}} a_{je} - \dist{m}{j}\right) \\
     & \phantom{\geq} - 2 \sum_{j\in B_{\ell}}  \dist{m}{j}.
  \end{align*}
  Denote \( 2 \sum_{j \in B_{\ell}} \dist{m}{j} \) by \( \lambda \).
  The inequality above implies in the following two inequalities:
  \begin{align}
    \left(\sum_{e\in s_{\ell}} a_{\ell e} - \dist{m}{\ell}\right)
     & \leq \frac{f_{m}^{S} + \lambda }{|B_{\ell}| + 1 } \label{inequality:proof-large-configuration-feasible-one} \\
    \left(\sum_{e\in s_{\ell}} a_{\ell e} - \dist{m}{\ell}\right) + \sum_{j\in A_{\ell}} \left(\sum_{e\in s_{j}} a_{je} - \dist{m}{j} \right)
     & \leq f_{m}^{S} + \lambda \label{inequality:proof-large-configuration-feasible-two}
    .
  \end{align}

  We model the task to bound \(X := \sum_{r} \left( \sum_{e\in s_{r}} a_{re}-\dist{m}{r}\right)_{+} \) by solving the problem of covering all the \( \left( \sum_{e\in s_{r}} a_{re} - \dist{m}{r} \right)_{+} \), given an instance of \( c \)-ordered covering.
  The idea behind this is the following.
  Each time we cover an element \( \left( \sum_{e\in s_{r}} a_{re} - \dist{m}{r} \right)_{+} \) of \( X \), we do so by applying either \cref{inequality:proof-large-configuration-feasible-one} or \cref{inequality:proof-large-configuration-feasible-two}.
  In case of \cref{inequality:proof-large-configuration-feasible-one}, we remove one element \( \left( \sum_{e\in s_{r}} a_{re} - \dist{m}{r} \right)_{+} \) from \( X \) and add a weight of \( \frac{f_{m}^{S} + \lambda }{|B_{\ell}| + 1 } \).
  In the other case of \cref{inequality:proof-large-configuration-feasible-two}, we remove multiple elements from the sum of \( X \) and add a weight of \( f_{m}^{S} + \lambda \).
  We ask ourselves how much weight is achieved when removing every element of \( X \).
  The resulting weight then directly represents an upper bound for \( X \).

  Next, we define an instance of \( c \)-ordered covering based on inequalities (1) and (2).
  Our instance is as follows:
  The elements are \( 1,\dots,p \).
  Consider element \( i \).
  It represents \( \left( \sum_{e\in s_{r}} a_{re} - \dist{m}{r} \right)_{+} \) of the \( i \)-th arriving request \( r \) of \( R'(\sigma) \).
  The sets \( B_{i} \) and \( A_{i} \) are given by the \( B_{r} \) and \( A_{r} \) as defined above.
  The parameter \( c \) of our ordered covering is \( f_{m}^{S} + \lambda \).
  For every element \( i \) there is a set \( \{i\} \) of weight \( c/(|B_{i}|+1) \) and a set \( \{i\} \cup A_{i} \) of weight \( c \).
  Notice that the weights of the sets correspond to \cref{inequality:proof-large-configuration-feasible-one} and \cref{inequality:proof-large-configuration-feasible-two}, respectively.

  Now, we show that this is a proper \( c \)-ordered covering instance.
  For any element \( i \), \( A_{i} \cap B_{i} = \emptyset \) by definition and \( A_{i} \cup B_{i} = \{1,\dots,i-1\} \), because exactly the requests of \( R'(\sigma) \) that arrived earlier than the request corresponding to \( i \) have a defined value for \( \min\{ \sum_{e\in s_{j}} a_{j e}, \dist{\hat{F}}{j} \} \).
  If a request \( r \) is in some \( A_{i} \) and in \( B_{i+1} \), it contributed to building a large facility of which the distance to itself is less than \( a_{r} \).
  Thus, for all following elements \( j > i \), \( r \) will stay in \( B_{j} \).
  In other words, for any two elements \( i, j \) with \( i < j \) it holds that \( B_{i} \subseteq B_{j} \).

  By \cref{corollary:PD:Subproblem:covering-works-with-some-weight} we know that \( \left\{\left( \sum_{e \in s_{r}} a_{re} - \dist{m}{r} \right)_{+} \,|\, r\in R'\right\} \) can be covered with a total weight of \( 2 (f_{m}^{S} + \lambda) H_{n} \).
  Each time an element is covered, this corresponds to applying either \cref{inequality:proof-large-configuration-feasible-one} or \cref{inequality:proof-large-configuration-feasible-two} to the respective \( \left( \sum_{e\in s_{r}} a_{re} - \dist{m}{r} \right)_{+} \) term, indicated by the increase in the weight of the covering.
  Combined with Condition~\ref{inequality:main-assumption}, we conclude:
  \begin{align*}
     & \sum_{r\in R'} \left(\sum_{e\in s_{r}} a_{re} - \dist{m}{r}\right)_{+}
    \leq 2 (f_{m}^{S} + \lambda) H_{n}                                                    \\
     & \phantom{\Rightarrow} \leq 2 f_{m}^{S} H_{n} + 4 H_{n} \sum_{r \in R'} \dist{m}{r} \\
     & \Rightarrow
    \sum_{r\in R'} \left(\sum_{e\in s_{r}} a_{re} - 5 H_{n} \dist{m}{r} \right)_{+}
    \leq 2 f_{m}^{S} H_{n}
    \leq 2 \frac{|S|}{|\sigma|} H_{n} f_{m}^{\sigma}                                      \\
     & \Rightarrow
    \sum_{r\in R'} \left( \sum_{e\in s_{r}} a_{re} - 5 H_{n} \dist{m}{r}\right)_{+}
    \leq 2 \sqrt{|S|} H_{n} f_{m}^{\sigma}                                                \\
     & \Rightarrow
    \sum_{r\in R'} \left(\sum_{e\in s_{r}\cap \sigma}\frac{a_{re}}{5 \sqrt{|S|} H_{n}} - \frac{\dist{m}{r}}{\sqrt{|S|}}\right)_{+}
    \leq \frac{2}{5} f_{m}^{\sigma}                                                       \\
     & \Rightarrow
    \sum_{r\in R'} \left(\sum_{e\in s_{r}\cap \sigma}\gamma a_{re} - \dist{m}{r} \right)_{+}
    \leq f_{m}^{\sigma}
    .
  \end{align*}
\end{proof}

By \cref{lemma:PD-extend:feasability-for-medium-configurations} and \cref{lemma:PD-extend:feasability-for-configurations} we conclude that the following corollary holds.

\begin{corollary}\label{corollary:PD-extend:gamma-ar-feasable}
  The dual variables \( a_{re} \) scaled by \( \gamma \) provide a feasible dual solution.
\end{corollary}

\begin{proof}[Proof of \cref{theorem:PD-extend:algorithm-competitive}]
  By \cref{corollary:PD-extend:gamma-ar-feasable}, the dual variables \( a_{re} \) scaled by \( \gamma = 1/(5\sqrt{|S|} \, H_{n}) \) provide a feasible solution to the dual and thus \( \sum_{r\in R} \sum_{e\in s_{r}} a_{re} \leq 5\sqrt{|S|}\,H_{n} \cdot \opt{} \) due to weak duality.
  By \cref{corollary:PD-extend:cost-of-algorithm-bounded-by-dual-vars}, the cost of \DetAlg{}'s solution is at most
  \begin{align*}
    3 \sum_{r\in R} \sum_{e\in s_{r}} a_{re} \leq 15\sqrt{|S|}\,H_{n} \cdot \opt{}
    .
  \end{align*}
\end{proof}

\subsection{Improved Bounds}\label{section:deterministic-algorithm:Improved-Bounds}

In the following subsection, we show how we can derive better bounds on the competitive ratio of \DetAlg{} when the cost function \( f_{m}^{\sigma} \) is restricted.
Complementarily, we also derive an adaptive lower bound for the chosen restriction.
Both results show that the competitive ratio heavily depends on the given construction cost function, both in the lower and the upper bound.

Assume the facility cost is equal for all points \( m \in M \) and depends only on the size of the configuration, i.e., we can write the cost function as \( g(|\sigma|) = f_{m}^{\sigma} \).
We consider the class of functions
\begin{align*}
  \mathcal{C}= \left\{ g_{x}(|\sigma|) = |\sigma|^{\frac{x}{2}}\,\big|\, x \in [0,2] \right\}
  .
\end{align*}

Observe that \( \mathcal{C} \) intuitively contains functions that behave as the root function varying between a constant (\( x=0 \)) and a linear function (\( x = 2 \)).
It seems natural that costs for more commodities increase smoothly while the function is subadditive.

Our results concerning cost functions of class \( \mathcal{C} \) are summarized in \cref{theorem:deterministic-algorithm:improved-bound}.
Its proof is given below, first for the upper bound and afterwards for the lower bound.
Before going to the proof, consider some examples that are implied by \cref{theorem:deterministic-algorithm:improved-bound} in which our algorithm actually achieves a tight competitive ratio concerning the part depending on \( |S| \).
A linear function, i.e., \( x=2 \), yields an upper bound for \DetAlg{} of \( \mathcal{O} ( \log n ) \) and a general lower bound of \( \Omega (\frac{\log n}{\log \log n}) \).
Note that for this function, \opt{} has no advantage by combining commodities in a single facility and prediction is essentially useless.
Our algorithm achieves the tight bound (concerning \( |S| \)) by roughly mimicking separate instances of the OFLP for each commodity.
The square root function, i.e., \( x = 1 \), yields an upper bound of \( \mathcal{O}( \sqrt[4]{|S|} \log n ) \) for \DetAlg{} and a general lower bound of \( \Omega (\sqrt[4]{|S|} + \frac{\log n}{\log \log n}) \).
Trivially, setting \( x = 0 \) removes the necessity of distinguishing between small and large facilities, yielding the upper bound \( \mathcal{O} (\log n) \) for \DetAlg{} and the lower bound \( \Omega (\frac{\log n}{\log \log n}) \) identical to the OFLP.
When considering only the term depending on \( |S| \), our algorithm's competitive ratio comes close to the lower bound.
\cref{figure:improved-bounds:figure-for-competitive-ratio} sketches the respective terms for comparison.

\begin{figure}[htb]
  \centering
  \includegraphics[width=0.47\textwidth, page=4, trim= 0.25cm 11.5cm 20cm 0.25cm]{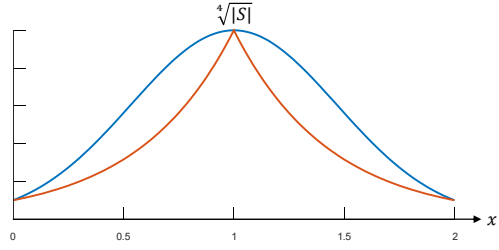}
  \caption{A depiction of the functions \( \SqrtS^{(2\,x-x^{2})/2} \) (blue) and \( \min\{\SqrtS^{(2-x)/2},\SqrtS^{x/2}\} \) (orange) using \( |S| = 10.000 \). For \( x\in\{0,1,2\} \) the functions have the same value: i.e., the part depending on \( |S| \) is equal in the lower and the upper bound. They both have a peak of value \( \sqrt[4]{|S|} \) at \( x = 1 \).}
  \label{figure:improved-bounds:figure-for-competitive-ratio}
\end{figure}

\begin{theorem}\label{theorem:deterministic-algorithm:improved-bound}
  Fix a cost function \( g_{x} \in \mathcal{C} \).
  \textnormal{\DetAlg{}} achieves a competitive ratio of \( \BigO{\SqrtS^{\frac{2\,x-x^{2}}{2}} \log n} \).
  No randomized online algorithm for the OMFLP can achieve a competitive ratio better than \( \Omega \left( \min\{\SqrtS^{\frac{2-x}{2}},\SqrtS^{\frac{x}{2}}\} + \lln \right) \).
\end{theorem}

\subsubsection{Proof of the upper bound}\(  \)

\noindent
Fix \( g_{x} \in \mathcal{C} \).
When considering our analysis, we observe that we essentially distinguish between configurations of a size of at most \( a \) (see \cref{lemma:PD-extend:feasability-for-medium-configurations}) and those of a size of at least \( a \) (see \cref{lemma:PD-extend:feasability-for-configurations}), where \( a \) is some threshold.
In our analysis for the general case, this threshold is \( \SqrtS \).
However, it can be optimized when having knowledge about the construction cost function as follows.

Let \( \sigma \) be a configuration of the size \( |\sigma| \leq a \).
Then we immediately get a scaling factor of \( \mathcal{O}((\frac{|\sigma|}{g_{x}(|\sigma|)}\log n)^{-1}) \) to reach dual feasibility.
Observe that \( \frac{|\sigma|}{g_{x}(|\sigma|)}=|\sigma|^{(1-(x/2))} \leq a^{(1-(x/2))} = \frac{a}{g_{x}(a)}\).
Thus, for such a configuration, the scaling factor is \( \mathcal{O}((\frac{a}{g_{x}(a)}\log n)^{-1}) \).
For configurations of a size of at least \( a \), we end up with a scaling factor of \( \mathcal{O}((\frac{g_{x}(|S|)}{g_{x}(a)}\log n)^{-1}) \).
The competitive ratio of our algorithm is thus in general given by \( \mathcal{O}(\max\{\frac{a}{g_{x}(a)},\frac{g_{x}(|S|)}{g_{x}(a)} \}\log n) \).
We set \( \frac{a}{g_{x}(a)} = \frac{g_{x}(|S|)}{g_{x}(a)} \) and solve for \( a \).
This yields \( a =  g_{x}(|S|) = \SqrtS^{x} \) for our threshold value.
Plugging in \( a \) in the competitive ratio yields the bound of \cref{theorem:deterministic-algorithm:improved-bound}.

\subsubsection{Proof of the lower bound}\(  \)

\noindent
Next, we turn our attention to a lower bound for functions in \( \mathcal{C} \) that are parametrized in \( x \).
Consider the construction used the proof of the lower bound in \cref{theorem:lower-bound}.

Fix \( g_{x} \in \mathcal{C} \).
Independent of the cost function, we concluded in \cref{inequality:lower-bound-expectation-in-S} that if \alg{} does not proceed \( \SqrtS/2 \) rounds it has to cover expectedly \( \E[T] \geq |S|/16 \) commodities.
In the former case, \alg{} pays at least \( (\SqrtS/2)\, g(1) = \SqrtS/2 \).
In the latter case, \alg{} pays at least
\begin{align}
  g_{x}(\E[T])
  \geq g_{x}\left(\frac{|S|}{16}\right)
  = \left(\frac{|S|}{16}\right)^{\frac{x}{2}}
  \geq \frac{\SqrtS^{x}}{16}
  .
\end{align}
Thus, \alg{}'s expected cost is at least
\begin{align*}
  \frac{1}{16}\min\{\SqrtS, \SqrtS^{x}\}
  .
\end{align*}
\opt{} pays at most \( g_{x}(\SqrtS) = \sqrt{|S|}^{\frac{x}{2}} \).
Thus, the resulting competitive ratio is at least
\begin{align*}
  \frac{\frac{1}{16}\,\min\{\SqrtS,\SqrtS^{x}\}}{\sqrt{|S|}^{\frac{x}{2}}}
  = \frac{1}{16}\,\min\{\SqrtS^{\frac{2-x}{2}},\SqrtS^{\frac{x}{2}}\}
  .
\end{align*}

Recapitulate that \cref{theorem:lower-bound} holds for a single point.
Thus, we can extend the bound above by using Fotakis' lower bound \cite{Fotakis2008} to the bound stated in the theorem.

\section{A randomized Algorithm}\label{section:Randomized-Algorithm}

In the next section we present a randomized algorithm for the OMFLP we call \RandAlg{}.
Randomization has the advantage that the decision process is highly efficient in comparison to a deterministic approach.
While low computational complexity in a time step is not of interest when considering the competitive ratio, it might be useful in practice.
Additionally, we are able to prove a slightly better competitive ratio of \( \mathcal{O}(\SqrtS \lln) \).

\subsection{Algorithm}\label{section:Randomized-Algorithm:Algorithm}

Our algorithm \RandAlg{} is inspired by Meyerson's randomized algorithm for the OFLP \cite{Meyerson2005} achieving an expected competitive ratio of \( \mathcal{O}(\lln) \).
We again consider the metric non-uniform OMFLP.
Similar to Meyerson, we first introduce classes for facility costs to deal with the non-uniformity.

Fix a configuration \( \sigma \).
Consider the set of all possible different \( f_{m}^{\sigma} \) rounded down to the nearest power of \( 2 \) in increasing order \( C_{1}^{\sigma},\dots, C_{n}^{\sigma} \).
We call \( C_{i}^{\sigma} \) the \emph{class} \( i \) with respect to \( \sigma \) representing a facility cost for \( \sigma \) occurring at a set of points in \( M \).
Observe that for any \( i \geq 1 \) it holds that \( 2\,C_{i}^{\sigma} \leq C_{(i+1)}^{\sigma} \).
Let \( \dist{C_{i}^{\sigma}}{m} \) denote the minimal distance from a point \( m\in M \) to a point in class \( i \) for \( \sigma \).
By rounding down the facility costs, the competitive ratio of our algorithm increases by at most a factor of \( 2 \).

Our main focus for randomization is to define probabilities for the construction of small and large facilities when a request arrives.
To simplify our analysis, we define the probabilities in a way that the expected costs that are paid for a request for connecting to a facility, for the construction of small facilities and for the construction of a large facility are equal.
When considering Meyerson's algorithm \cite{Meyerson2005}, we observe that the probability for constructing a facility when a request \( r \) arrives depends on the cost \( r \) creates under the assumption that it is simply connected.
The connection cost can be interpreted as the budget of \( r \).

\begin{figure}[htb]
  \centering
  \includegraphics[width=0.47\textwidth, clip=true, page=3, trim=1.75cm 11.5cm 16.25cm 2.25cm]{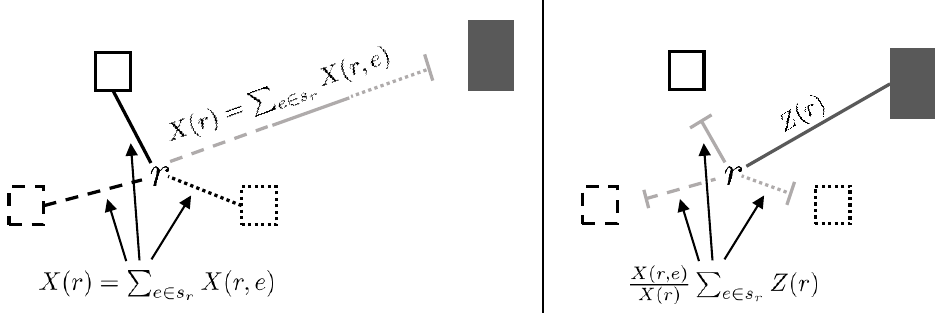}
  \caption{Assume that \( r \) gets connected and requests \( 3 \) commodities (indicated by normal, dotted and dashed lines). On the left: It is cheapest for \( r \) to connect to small facilities (not filled rectangles) only. The large facility (filled rectangle) is too far away. On the right: It is cheapest to connect to the nearest open large facility. We distribute a share of \( \frac{X(r,e)}{X(r)} \) of the total connection cost of \( r \) to every commodity.}
  \label{figure:rand-algo:example}
\end{figure}

Assume that a request \( r \) arrives and demands a set of commodities \( s_{r} \subseteq S \).
Either it is connected only to small facilities or it is connected to a single large facility.
\cref{figure:rand-algo:example} depicts the two possible situations.
In the former case, for any commodity \( e \in s_{r} \), either \( r \) is connected to the closest facility serving \( e \), or it is cheaper to open a new small facility of some class and connect \( r \) to it.
Thus, the connection cost is bounded by
\begin{align*}
  X(r,e) & := \min \{ \dist{F(e)}{r}, \min_{\,i} \{C_{i}^{\{e\}} + \dist{C_{i}^{\{e\}}}{r} \} \}
  .
\end{align*}
So, the total connection cost if the request is connected to only small facilities is
\begin{align*}
  X(r) & := \sum_{e\in s_{r}} X(r,e)
  .
\end{align*}
In the latter case, the request is either connected to a single large facility that exists, or it is cheaper to construct a new large facility of some class and connect \( r \) to it, such that the cost is at most
\begin{align*}
  Z(r) & := \min \{\dist{\hat{F}}{r}, \min_{\,i} \{C_{i}^{S} + \dist{C_{i}^{S}}{r}\} \}
  .
\end{align*}
So, assuming that the request gets connected, the total expected assignment cost can be bounded by exactly
\begin{align*}
  \E[r,asg] = \min\{ X(r),Z(r) \}
  .
\end{align*}

Our algorithm \RandAlg{} is depicted in \cref{algorithm:randomized-algorithm}.
It ensures that on expectation the total cost for all small facilities and the total cost for large facilities of a request equals \( \min\{ X(r),Z(r) \} \).

For large facilities, the expected assignment cost of \( r \) is distributed over all classes similar to Meyerson's approach for non-uniform facility costs \cite{Meyerson2005}.
Thereby, class \( i \) receives a portion of \( \E[r,asg] \) proportional to the improvement for \( r \) if there were a facility of class \( i \): i.e., \( \dist{C_{(i-1)}^{S}}{r} - \dist{C_{i}^{S}}{r} \).
This portion is divided for \( i \) by the construction cost of a facility of class \( i \).

For small facilities this is done similarly.
However, the probabilities for a single commodity \( e \in s_{r} \) are scaled by \( X(r,e)/X(r) \), which represents \( e \)'s share on the total value of \( \E[r,asg] \).

\begin{algorithm}[htb]
  \caption{\RandAlg{} on arrival of request \( r \) with set \( s_{r} \)}
  \label{algorithm:randomized-algorithm}
  \begin{algorithmic}[1]
    \ForAll{Classes \( i \)}
    \ForAll{\( e\in s_{r} \)}
    \State Build a small facility of class \( i \) in
    \Statex \hspace*{1.2cm} configuration \( \{ e \} \) closest to \( r \) with
    \Statex \hspace*{1.4cm} \( \Prob[r,e,i] = \frac{\dist{C_{(i-1)}^{\{e\}}}{r} - \dist{C_{i}^{\{e\}}}{r}}{C_{i}^{\{e\}} } \cdot \frac{X(r,e)}{X(r)} \)
    \Statex \hspace*{1.2cm} where \( \dist{C_{0}^{\{e\}}}{r} := \min \{Z(r), X(r)\} \).
    \EndFor
    \State Build a large facility of class \( i \) closest to \( r \) with
    \Statex \hspace*{0.9cm} \( \Prob[r,S,i] = \frac{ \dist{C_{(i-1)}^{S}}{r} - \dist{C_{i}^{S}}{r} }{C_{i}^{S}} \)
    \Statex \hspace*{0.75cm} where \( \dist{C_{0}^{S}}{r} := \min \{Z(r), X(r)\} \).
    \EndFor
  \end{algorithmic}
\end{algorithm}

\subsection{Analysis}\label{section:Randomized-Algorithm:Analysis}

For the analysis, we have the following outline.
The analysis is closely related to the analysis of \cite{Meyerson2005} but we need to carefully handle the size of a configuration.
We analyze the algorithm's cost for all requests of a fixed \emph{optimal center}.
An optimal center is a facility placed by \opt{} together with all the commodities connected to it in the offline solution.
We start by showing that the expected cost for connecting a request and for the construction of small and large facilities due to a request is equal in \cref{lemma:randomized:expected-cost-is-equal}.
This holds, when considering the total commodities of a request as well as when considering only a single one.
Afterwards, we establish \cref{lemma:randomized:expected-cost-per-commodity-total} saying that when considering either a single commodity of an optimal center or all commodities of \( S \), the algorithm's cost is bounded in the cost for the construction of a facility serving the single commodity or all commodities at the point of the optimal center (within a factor of \( \mathcal{O}(\lln) \)).
Finally, the key idea of our proof is a distinction between \emph{small} optimal centers having a configuration containing at most \( \SqrtS \) commodities (\cref{lemma:randomized:competitive-ratio-small-centers}) and \emph{large} optimal centers with a configuration containing more than \( \SqrtS \) commodities (\cref{lemma:randomized:competitive-ratio-large-centers}).
This step is similar to the end of the analysis of \DetAlg{} in \cref{section:deterministic-algorithm:Analysis}.
Using the last two lemmas, we close our section with a small proof of the theorem below.

\begin{restatable}{theorem}{theoremRandomizeAlgorithmCompetitive}
  \label{theorem:randomized:competitive-ratio}
  \textnormal{\RandAlg{}} has a competitive ratio of \begin{align*}
    \BigO{\SqrtS \cdot \lln}
    .
  \end{align*}
\end{restatable}

\begin{lemma}\label{lemma:randomized:expected-cost-is-equal}
  Consider a configuration \( \tau \in S \cup \{\{s\} \,|\, s\in S\} \).
  Let \( \E[r,\tau,asg] \) be the expected assignment cost of \( r \) charged to \( \tau \).
  Let \( \E[r,\tau,\ell F] \) be the expected cost for large facilities due to \( r \) charged to \( \tau \).
  Let \( \E[r,\tau, sF] \) be the expected cost for small facilities due to \( r \) charged to \( \tau \).

  Then
  \begin{align*}
    \E[r,\tau,asg] = \E[r,\tau,\ell F] = \E[r,\tau, sF].
  \end{align*}
\end{lemma}

\begin{proof}
  Fix \( \tau = S \) and \( r\in R \).
  We charge the total expected cost of \( r \) to \( S \).
  Observe that we already discussed that the total assignment cost of \( r \) is \( \E[r,S,asg] = \min \{X(r), Z(r)\} \).
  The expected construction cost for large facilities due to \( r \) charged to \( S \) is then
  \begin{align*}
    \E[r,S,\ell F]
     & = \sum_{i} \Pr[r,S,i] \, C_{i}^{S}
    = \sum_{i} (\dist{C_{(i-1)}^{S}}{r} - \dist{C_{i}^{S}}{r})
    \\
     & = C_{0}^{S}
    = \min \{Z(r), X(r)\}
    = \E[r,S,asg]
    .
  \end{align*}
  The expected construction cost for small facilities due to \( r \) charged to \( S \) is
  \begin{align*}
     & \E[r,S,sF]
    = \sum_{e\in s_{r}} \sum_{i} \Pr[r,e,i] \, C_{i}^{\{e\}}
    \\
     & = \sum_{e\in s_{r}} \sum_{i} \left( \left( \dist{C_{(i-1)}^{\{e\}}}{r} - \dist{C_{i}^{\{e\}}}{r} \right) \, \frac{X(r,e)}{X(r)} \right)
    \\
     & = \sum_{e\in s_{r}} \frac{X(r,e) \, C_{0}^{\{e\}}}{X(r)}
    = \sum_{e\in s_{r}} \frac{X(r,e) \, \min\{Z(r), X(r)\}}{X(r)}
    \\
     & = \min\{Z(r), X(r)\}
    = \E[r,S,asg]
    .
  \end{align*}

  Thus, for \( \tau = S \) the lemma holds true.
  Next, each commodity \( e \in s_{r} \) receives a share of the assignment cost of \( r \) of the size
  \begin{align*}
    \E[r,\{e\}, asg] = \min\{Z(r),X(r)\} \cdot \frac{X(r,e)}{X(r)}
    .
  \end{align*}
  Then the total expected assignment cost of \( r \) is split up entirely onto its commodities: i.e.,
  \begin{align*}
    \sum_{e\in s_{r}} \E[r,\{e\}, asg]
     & = \sum_{e\in s_{r}} \min\{Z(r),X(r)\} \cdot \frac{X(r,e)}{X(r)}
    \\
     & = \min\{Z(r),X(r)\}
    = \E[r,asg]
    .
  \end{align*}
  We distribute the total expected construction cost for large facilities in the same way to the commodities of \( r \) and it holds that \( \E[r,\{e\}, \ell F] = \E[r,\{e\}, asg] \).
  Consider the expected construction cost for small facilities of commodity \( e \):
  \begin{align*}
     & \E[r,\{e\}, sF]
    = \sum_{i} \Pr[r,e,i] \, C_{i}^{\{e\}}
    \\
     & = \sum_{i} \left( \dist{C_{(i-1)}^{\{e\}}}{r} - \dist{C_{i}^{\{e\}}}{r}{C_{i}^{\{e\}}}\right)\frac{X(r,e)}{X(r)}
    \\
     & = C_{0}^{\{e\}}\frac{X(r,e)}{X(r)}
    = \min\{Z(r),X(r)\}\frac{X(r,e)}{X(r)}
    = \E[r,\{e\}, asg]
    .
  \end{align*}
  For any \( e\notin s_{r} \), the respective expected costs are simply zero.
  Also it follows that summing up over all commodities of \( s_{r} \) yields the total expected cost of \( r \).
  Thus, for any \( e \in S \), \( \E[r,\{e\},asg] = \E[r,\{e\}, \ell F] = \E[r,\{e\}, sF] \).
\end{proof}

\begin{lemma}
  \label{lemma:randomized:expected-cost-per-commodity-total}
  Consider an optimal center \( c \) and let \( R_{c} \) be the set of requests connected to \( c \) by \opt{}.
  Fix a configuration \( \tau \in S \cup \{ \{s\} \,|\, s\in S \} \).
  Let \( i \) be the class of \( c \) with respect to \( \tau \)
  The expected total cost charged to commodities in \( \tau \) due to requests in \( R_{c} \) is at most \( \BigO{\lln \,(C_{i}^{\tau} + \sum_{r\in R_{c}} \dist{r}{c})} \).
\end{lemma}

\begin{proof}
  We will analyze requests based on their distance to the center.
  For this, we divide the request set of a center as follows:
  Let \( B_{\alpha} \) be the set of points in a distance of \( [t^{\alpha - 1} \overline{\te{Asg}_{c}}, t^{\alpha} \overline{\te{Asg}_{c}}] \) where \( t = \frac{\log n}{\log \log n} \) and \( \overline{\te{Asg}_{c}} = \frac{\sum_{r \in R_{c}} \dist{r}{c} }{|R_{c}|} \) is the average assignment cost of \opt{} for \( c \).
  Observe that \( B_{t+1} \) is empty since \( t^{t+1} > n \).
  A point in \( B_{t+1} \) would incur an assignment cost for \( c \) higher than the assignment cost of all requests in \( R_{c} \).

  For a fixed \( B_{\alpha} \), observe that the distance of any request \( r \) to a point of some class \( j \) concerning \( \tau \) is bounded by \( \dist{C_{j}^{\tau}}{r} \leq \dist{C_{j}^{\tau}}{c} + t^{\alpha} \, \overline{\te{Asg}_{c}} \).
  We say event \( j \) occurs if a facility opens in a distance of \( \dist{C_{j}^{\tau}}{c} + 2 \, t^{\alpha} \, \overline{\te{Asg}_{c}} \) of \( c \).
  Any request of \( B_{\alpha} \) that opens a facility of class \( j \) or higher in configuration \( \tau \) will cause event \( j \) to happen.
  Let \( \delta \) be a constant.
  We say event \( \delta^{*} \) occurs if a facility in distance of \( \delta \, t^{\alpha} \, \overline{\te{Asg}_{c}}  \) to \( c \) opens.
  Observe that for \( j \geq i \), \( \dist{C_{j}^{\tau}}{c} = 0 \) and it automatically follows that any facility built of class \( j \geq i \) triggers both events \( j \) and \( \delta^{*} \).

  Consider a fixed set \( B_{\alpha} \) and let \( B_{\alpha}^{\tau} \subseteq B_{\alpha} \) be the set of requests in \( B_{\alpha} \) that request a commodity in \( \tau \).
  Let \( B_{\alpha}^{\tau}(x) \) be the set of requests in \( B_{\alpha}^{\tau} \) that appear before an event \( x \) has happened.

  Consider the set \( B_{\alpha}^{\tau}(\delta ^{*}) \).
  We next charge each class \( j \) an expected assignment cost of
  \begin{align*}
    \E[B_{\alpha},r,\tau,j,asg] =\dist{C_{(j-1)}^{\tau}}{r} - \dist{C_{j}^{\tau}}{r},
  \end{align*}
  where \( \dist{C_{0}^{\tau}}{r} := \E[B_{\alpha},r,\tau,asg] \).
  Observe \( \sum_{j} \E[B_{\alpha},r,\tau,j,asg] = \E[B_{\alpha},r,\tau,asg] \).
  If \( \tau = \{s\} \) for some \( s\in S \), then \( \E[B_{\alpha},r,\tau,asg] \) corresponds to the assignment cost charged to commodity \( s \).
  If \( \tau = S \), \( \E[B_{\alpha},r,\tau,asg] \) corresponds to the assignment cost of \( r \).

  Consider event \( j \).
  The expected total assignment cost of class \( j \) considering \( \tau \) over all requests \( r \in B_{\alpha}^{\tau}(j) \) is exactly
  \begin{align*}
    \sum_{r \in B_{\alpha}^{\tau}(j)} \E[B_{\alpha},r,\tau,j,asg] = C_{j}^{\tau}
    .
  \end{align*}
  After this expected assignment cost the algorithm builds either a small facility of class \( j \) or a large facility of class \( j \) serving \( \tau \) and event \( j \) happens.

  Assume that \( j \) occurred, but \( \delta ^{*} \) has not.
  Consider a request \( r \in B_{\alpha}^{\tau}(\delta ^{*}) \setminus B_{\alpha}^{\tau}(j) \).
  Let
  \begin{align*}
    Y(r,\tau):=\begin{cases}
      X(r,s) & \te{if } \tau = {s} \te{ for } s\in S \\
      Z(r)   & \te{if } \tau = S
      .
    \end{cases}
  \end{align*}
  Observe that \( Y(r,\tau) \leq \dist{F(\tau)}{r} \) in any case.
  Since event \( j \) occurred, we know that
  \begin{align*}
     & \E[B_{\alpha},r,\tau,asg\,|\,r\in B_{\alpha}^{\tau}(\delta ^{*})\setminus B_{\alpha}^{\tau}(j)] \leq Y(r,\tau) \leq \dist{F(\tau)}{r} \\
     & \leq \dist{C_{i}^{\tau}}{r} + 3\, t^{\alpha}\, \overline{\te{Asg}_{c}}.
  \end{align*}
  In a distance of at most \( \dist{C_{j}^{\tau}}{r} \) to \( r \), there is a point of class \( j \).
  This point is in a distance of at most \( \dist{C_{j}^{\tau}}{r} + t^{\alpha} \, \overline{\te{Asg}_{c}} \) to \( c \).
  Therefore,
  \begin{align*}
    \dist{C_{j}^{\tau}}{r} + t^{\alpha} \, \overline{\te{Asg}_{c}}
    \geq \dist{C_{j}^{\tau}}{c}
    \Leftrightarrow
    \dist{C_{j}^{\tau}}{r}
    \geq \dist{C_{j}^{\tau}}{c} - t^{\alpha} \, \overline{\te{Asg}_{c}}
    .
  \end{align*}
  Additionally, \( \dist{C_{j}^{\tau}}{c} + 2\, t^{\alpha} \, \overline{\te{Asg}_{c}} \geq \delta \, t^{\alpha} \, \overline{\te{Asg}_{c}} \Leftrightarrow \dist{C_{j}^{\tau}}{c} \geq (\delta - 2) \, t^{\alpha} \, \overline{\te{Asg}_{c}} \), because \( \delta ^{*} \) has not occurred yet.
  Now consider
  \begin{align*}
     & \frac{\dist{C_{j}^{\tau}}{r}}{\E[B_{\alpha},r,\tau,asg\,|\,r\in B_{\alpha}^{\tau}(\delta ^{*})\setminus B_{\alpha}^{\tau}(j)]} \geq \frac{\dist{C_{j}^{\tau}}{c} - t^{\alpha}\, \overline{\te{Asg}_{c}}}{\dist{C_{j}^{\tau}}{c} + 3\, t^{\alpha}\, \overline{\te{Asg}_{c}}}
    \\
     & \geq \frac{\delta - 3}{\delta + 1}
    \\
     & \Leftrightarrow
    \dist{C_{j}^{\tau}}{r}
    \geq \frac{\delta - 3}{\delta + 1} \,\E[B_{\alpha},r,\tau,asg\,|\,r\in B_{\alpha}^{\tau}(\delta ^{*})\setminus B_{\alpha}^{\tau}(j)]
    .
  \end{align*}
  Assume that event \( j \) has happend.
  How much expected assignment cost is accumulated until event \( (j+1) \) has happened as well?
  We know that the expected assignment cost charged to classes higher than \( j \) is \( C_{(j+1)}^{\tau} \).
  Therefore,
  \begin{align*}
     &                                                        & \sum_{r\in B_{\alpha}^{\tau}(j+1)\setminus B_{\alpha}^{\tau}(j)} \dist{C_{j}^{\tau}}{r}
     & = C_{(j+1)}^{\tau}
    \\
     & \Leftrightarrow                                        &
    \sum_{r\in B_{\alpha}^{\tau}(j+1)\setminus B_{\alpha}^{\tau}(j)} \frac{\delta - 3}{\delta + 1} \, \E[B_{\alpha},r,\tau,asg]
     & \leq C_{(j+1)}^{\tau}
    \\
     & \Leftrightarrow                                        &
    \sum_{r\in B_{\alpha}^{\tau}(j+1)\setminus B_{\alpha}^{\tau}(j)} \E[B_{\alpha},r,\tau,asg]
     & \leq \frac{\delta + 1}{\delta - 3} \, C_{(j+1)}^{\tau}
    .
  \end{align*}
  Thus, the total expected assignment cost charged to configuration \( \tau \) until \( \delta ^{*} \) happens is
  \begin{align*}
    \sum_{r\in B_{\alpha}^{\tau}(\delta ^{*})} \E[B_{\alpha},r,\tau,asg]
     & \leq C_{1}^{\tau} + \sum_{j=1}^{i-1} \sum_{r \in B_{\alpha}^{\tau}(j+1) \setminus B_{\alpha}^{\tau}(j)} \E[B_{\alpha},r,\tau,asg]
    \\
     & \leq \sum_{j=1}^{i} \frac{\delta + 1}{\delta - 3} \,C_{j}^{\tau}
    \leq \frac{\delta + 1}{\delta - 3} \, 2 \, C_{i}^{\tau}
    .
  \end{align*}
  Since \( \E[r,\tau, asg] = \E[r,\tau,sF] = \E[r,\tau,\ell F] \), the total expected cost charged to \( \tau \) until \( \delta ^{*} \) is
  \begin{align*}
    \sum_{r\in B_{\alpha}^{\tau}(\delta ^{*})} \E[B_{\alpha}, r,\tau]
    \leq 6\, \frac{\delta + 1}{\delta - 3} \, C_{i}^{\tau}
    .
  \end{align*}
  After \( \delta ^{*} \) has happened, there is a facility close by for future requests for \( \tau \).
  We distinguish here between \( \alpha > 0 \) and \( \alpha = 0 \).
  Assume that \( \alpha > 0 \) and consider \( r\in B_{\alpha}^{\tau} \setminus B_{\alpha}^{\tau}(\delta ^{*}) \).
  Then
  \begin{align*}
     & \E[B_{\alpha},r,\tau \,|\,r\in B_{\alpha}^{\tau} \setminus B_{\alpha}^{\tau}(\delta ^{*})]
    \leq 3\, Y(r,\tau)
    \leq 3\,\dist{F(\tau)}{r}                                                                     \\
     & \leq 3\,(\dist{r}{c} + \delta \, t^{\alpha} \, \overline{\te{Asg}_{c}})
    \leq 3\,(\delta \, t + 1)\, \dist{r}{c}
    .
  \end{align*}
  Now consider \( r\in B_{0}^{\tau} \setminus B_{0}^{\tau}(\delta ^{*}) \).
  There is a facility in distance of at most \( \delta \, \overline{\te{Asg}_{c}} \) of \( c \) such that
  \begin{align*}
    \E[B_{0},r,\tau \, |\,r\in B_{0}^{\tau} \setminus B_{0}^{\tau}(\delta ^{*})]
     & \leq 3\, Y(r,\tau)
    \leq 3\, \dist{F(\tau)}{r}                           \\
     & \leq 3\, \dist{r}{c} + 3\,\overline{\te{Asg}_{c}}
    .
  \end{align*}
  Summing up over all \( B_{\alpha} \), the total expected cost charged to \( \tau \) is
  \begin{align*}
     & \sum_{\alpha=0}^{t} \sum_{r\in B_{\alpha}} \E[B_{\alpha},r,\tau]
    \leq 6\,\left(\lln + 1\right)\frac{\delta + 1}{\delta -3} \, C_{i}^{\tau}                                                                                         \\
     & \phantom{=} + 3 \, \left(\delta \lln + 1 \right) \sum_{r\in R_{c}} \dist{r}{c} + \sum_{r\in R_{c}} \delta \overline{\te{Asg}_{c}}
    \\
     & \leq 6\,\left(\lln + 1\right)\frac{\delta + 1}{\delta -3} \, C_{i}^{\tau} + 3 \, \left(\delta \lln + 1 + \frac{\delta}{3}\right) \sum_{r\in R_{c}} \dist{r}{c}
    .
  \end{align*}
  Since \( \delta \) is a constant, the lemma follows.
\end{proof}

\begin{lemma}
  \label{lemma:randomized:competitive-ratio-small-centers}
  Consider an optimal center \( c \) with configuration \( \sigma \) such that \( |\sigma| \leq \SqrtS \).
  Let \( R_{c} \) be the set of requests connected to \( c \) by \opt{}.
  The expected total cost of our algorithm paid for commodities of requests of \( R_{c} \) that are served by \( c \) in \opt{} is at most \( \BigO{\SqrtS\lln} \) times the cost of \opt{} concerning \( c \).
\end{lemma}

\begin{proof}
  We start with a single commodity \( s\in \sigma \).
  By \cref{lemma:randomized:expected-cost-per-commodity-total}, we know for the configuration \( \tau = \{s\} \) that the expected total cost charged to \( s \) is \( \BigO{\lln \,(C_{i}^{\{s\}} + \sum_{r\in R_{c}} \dist{r}{c})} \).
  Observe that for any \( s\in \sigma \) with class \( i \) at \( c \) it holds \( C_{i}^{\{s\}} \leq 2\,f_{c}^{\{s\}} \leq 2\,f_{c}^{\sigma} \) due to our rounding.
  Summing up over all \( |\sigma| \leq \SqrtS \) commodities, the total expected cost for \alg{} is \( \BigO{\SqrtS\,\lln \,(f_{c}^{\sigma} + \sum_{r\in R_{c}} \dist{r}{c})} \).
  \opt{} pays for \( c \) at least \( f_{c}^{\sigma} + \sum_{r\in R_{c}} \dist{r}{c} \) and the lemma is correct.
\end{proof}

\begin{lemma}
  \label{lemma:randomized:competitive-ratio-large-centers}
  Consider an optimal center \( c \) with configuration \( \sigma  \) such that \( |\sigma| > \SqrtS \).
  Let \( R_{c} \) be the set of requests connected to \( c \) by \opt{}.
  The expected total cost of our algorithm paid for requests of \( R_{c} \) that are served by \( c \) in \opt{} is at most \( \BigO{\SqrtS\lln} \) times the cost of \opt{} concerning \( c \).
\end{lemma}

\begin{proof}
  By \cref{lemma:randomized:expected-cost-per-commodity-total}, we know for the configuration \( \tau = S \) that the expected total cost charged to all requests in \( R_{c} \) is \\ \( \BigO{\lln \,(C_{i}^{S} + \sum_{r\in R_{c}} \dist{r}{c})} \).
  Observe that due to our rounding and due to Condition~\ref{inequality:main-assumption} it holds that \( C_{i}^{S} \leq 2\,f_{c}^{S} \leq \frac{|S|}{|\sigma|} \, f_{c}^{\sigma} \).
  Since \( |\sigma| > \SqrtS \), the total expected cost of \alg{} concerning \( c \) is thus \\ \( \BigO{\SqrtS\lln\,(f_{c}^{\sigma} + \sum_{r\in R_{c}} \dist{r}{c})} \).
  \opt{} pays for \( c \) at least \( f_{c}^{\sigma} + \sum_{r\in R_{c}} \dist{r}{c} \) and the lemma is correct.
\end{proof}

\begin{proof}[Proof of \cref{theorem:randomized:competitive-ratio}]
  Combining \cref{lemma:randomized:competitive-ratio-small-centers} and \cref{lemma:randomized:competitive-ratio-large-centers} for every optimal center proves \cref{theorem:randomized:competitive-ratio}.
\end{proof}

\section{Closing remarks}\label{section:outlook}

We considered a natural extension of the Facility Location Problem, introducing commodities in an online scenario.
A crucial property that is needed to have a competitive ratio sublinear in the number of commodities was the use of prediction.
We believe this is already an interesting insight because it poses the additional difficulty of \emph{how} and \emph{when} to predict.
Both our algorithms, the randomized as well as the deterministic one, achieved this by constructing a facility at some point that offers all possible commodities.
The crucial assumption we pose on the construction cost function to achieve this was Condition~\ref{inequality:main-assumption}.
It would be very interesting to know how the problem can be handled when dropping or loosening it.
The condition indirectly implies that the costs for single commodities are not too different: i.e., there is no commodity that somewhat results in a high increase in the construction cost when it is added to an existing configuration.
If the number of such \emph{heavy} commodities is small or even constant, it is simple to handle them.
Naturally, one could simply run our algorithms in which the heavy commodities are excluded such that a large facility becomes one including all non-heavy commodities.
This reflects the intuition that heavy commodities should be avoided as far as possible.
However, a nice algorithm that provably works for general cost functions is still an open problem.
Additionally, we lack a lower bound that exploits general cost functions as well.
Our lower bound utilizes a function depending only on the number of commodities.
This implies that Condition~\ref{inequality:main-assumption} holds for the lower bound as discussed in \cref{section:Introduction:Model-and-Problem-Definition}.
So, it is also open how the competitive ratio may change if general cost functions are allowed.
Regarding that in the offline case the approximability already changes depending on the assumptions on the cost function (see \cref{section:Introduction:Related-Work}), we expect it to be similar in the online case.

The visible asymptotic difference between our lower and upper bounds on the competitive ratio is the fact that the lower bound incorporates an \emph{additive} factor of \( \SqrtS \) while the upper bounds have a \emph{multiplicative} factor of \( \SqrtS \).
Of course it remains open if this gap can be closed, although we see similarities between our problem and the Online Facility Leasing Problem in which the bounds diverge similarly, see \cite{Nagarajan2013}.

In general, for many online problems that consider requests that have to be answered it is still open how they can be extended to incorporate the idea of heterogeneity.
Possibly, some of the properties we encountered are similar, such as the necessity of prediction.

\bibliography{bibliography}

\end{document}